\documentclass[11pt]{amsart}
\usepackage{amsmath, amssymb, amsthm}
\usepackage{enumerate}
\usepackage{booktabs}
\usepackage{bbold} 
\usepackage{mathtools}
\usepackage{comment}
\usepackage{geometry}
\usepackage{bm}
\usepackage{tikz}
\usetikzlibrary{arrows.meta, decorations.markings}
\usetikzlibrary{calc}
\usepackage[utf8]{inputenc}
\newtheorem{theorem}{Theorem}[section]
\newtheorem{lemma}[theorem]{Lemma}
\newtheorem{proposition}[theorem]{Proposition}

\theoremstyle{remark}
\newtheorem{definition}[theorem]{Definition}
\newtheorem{remark}[theorem]{Remark}
\newtheorem{example}[theorem]{Example}
\newtheorem{assumption}[theorem]{Assumption}

\newcommand{\bA}{\mathbf{A}}

\newcommand{\R}{\mathbb{R}}

\newcommand{\Res}{\operatorname{Res}}
\newcommand{\curl}{\operatorname{curl}}
\newcommand{\supp}{\operatorname{supp}}
\newcommand{\dist}{\operatorname{dist}}
\newcommand{\Rank}{\operatorname{rank}}
\newcommand{\Ran}{\operatorname{Ran}}
\newcommand{\Rea}{\operatorname{Re}}
\newcommand{\Ima}{\operatorname{Im}}
\newcommand{\cP}{\mathcal P}

\DeclareMathOperator{\Div}{div}
\DeclareMathOperator{\Curl}{curl}

\definecolor{RoyalBlue}{RGB}{25, 25, 112} 
\definecolor{DarkBrown}{RGB}{101, 67, 33}

\numberwithin{equation}{section}

\title{Semiclassical resonances under local magnetic fields}
\author[Pavel Exner]{Pavel Exner}
\address[P. Exner]{\newline Doppler Institute for Mathematical Physics and Applied Mathematics, 
Czech Technical University, Břehová 7, 11519 Prague, Czechia \\
Department of Theoretical Physics, NPI, Academy of Sciences, 25068 Řež near Prague, Czechia}
\email{exner@ujf.cas.cz}
\author[Ayman Kachmar]{Ayman Kachmar}
\address[A. Kachmar]{\newline Department of Mathematics \normalfont{and} PDE Research Unit--Center for Advanced Mathematical Sciences (CAMS), American University of Beirut,
Beirut, Lebanon.}
\email{ak292@aub.edu.lb}

\date{\today}
\begin{document}

\begin{abstract}
We study resonances for the semiclassical magnetic Laplacian in the full plane with a compactly supported magnetic field in the framework of semiclassical complex scaling and black box scattering theory. Assuming that the magnetic field is locally constant, we prove the existence of semiclassical resonances near the Landau levels with exponentially small imaginary parts. We also prove that resonances emerge from a magnetic step discontinuity along a curved interface or a
non-degenerate magnetic well, and in the vicinity of anharmonic Landau levels if the field has an isolated zero.
\end{abstract}



\maketitle


\section{Introduction}

The behavior of a quantum system and its classical counterpart is sometimes substantially different. As a simple example, consider a charged particle in two dimensions exposed to a constant magnetic field supported in a disk. Outside it the trajectory is a straight line; if it hits the disk in a non tangential way, it leaves  again and the time spent in the disk is  shorter the stronger the field is. Our aim is to show that, in contrast, the quantum particle can be trapped temporarily  in the field support and the lifetime of such a quasi-stationary state increases exponentially with the field strength. The effect is of a semiclassical nature and occurs for a wide range of local magnetic field configurations.

\subsection{Setting and main question}
We investigate the existence of resonances for the semiclassical magnetic Laplacian in $\mathbb R^2$, 
\begin{equation}\label{eq:def-ml}
P(h)=(-ih\nabla-\bA)^2
\end{equation}
under the following hypotheses on the vector potential $\bA$:
\begin{enumerate}
    \item $\bA\in H^1_{\rm loc}(\mathbb R^2) \cap L^\infty_{\rm loc}(\mathbb R^2)$;
    \item There exists $R_0>0$ such that $\bA=\bA_\alpha$ on $\mathbb R^2\setminus B(0,R_0)$, where $B(0,R_0)$ is the open disk of radius $R_0$ centered at $(0,0)$ and $\bA_\alpha$ is the Aharonov-Bohm potential with a flux $\alpha>0$,
\begin{equation}\label{eq:def-AB}
\bA_\alpha (x)=\Bigl(-\frac{\alpha x_2}{|x|^2},\frac{\alpha x_1}{|x|^2}\Bigr).
\end{equation}
\end{enumerate}
Our hypotheses ensure that the magnetic field $B=\Curl\bA$ is square integrable and with compact support. In particular,  $B$ vanishes outside the disk $B(0,R_0)$. Note that by Stokes' theorem, the flux $\alpha$ must satisfy the compatibility condition
\begin{equation}\label{eq:cond-alpha}
\alpha=\frac1{2\pi}\int_{B(0,R_0)}B(x)dx.
\end{equation}

As we have indicated, the motivation for our study of $P(h)$ arises from the question whether a strong local magnetic field can create long-living resonances. Considering the operator $(-i\nabla - b\bA)^2$ with a large parameter $b$, one can set $h = b^{-1}$ in which case the rescaled problem reduces to the analysis of semiclassical resonances for the operator $P(h)$.\medskip

\subsection{On the definition of the operator}

We now turn to the precise definition of the operator $P(h)$ starting from the quadratic form
\[
Q_{\bA,h}(u) = \int_{\R^2} |(-ih\nabla - \bA)u|^2 \, dx, \quad u \in C_c^\infty(\R^2),
\]
which is non-negative and closable. Its Friedrichs extension yields a self-adjoint operator (our operator $P(h))$  on $L^2(\R^2)$ with domain
\[\mathcal D=\{u\in L^2(\R^2)\colon (-ih\nabla-\bA)u\in L^2(\R^2;\mathbb C^2),(-ih\nabla-\bA)^2u\in L^2(\R^2)\},\]
where 
\[(-ih\nabla-\bA)^2=-h^2\Delta-2i h\bA\cdot\nabla-ih(\Div \bA) +|\bA|^2.\] 
The vector potential $\bA$ is divergence free outside $B(0,R_0)$ and decays at infinity like $1/|x|$. The domain of $P(h)$ is then independent of $h$ and is given by
\begin{equation}\label{eq:Domain-H}
\mathcal D=\{u\in H^1(\R^2)\colon \Delta u\in L^2(\R^2)\}=H^2(\R^2).
\end{equation}

Since the magnetic field is compactly supported and $P(h)\geq 0$ in the sense of quadratic forms, the spectrum of $P(h)$ is $[0,+\infty)$ and is purely essential. 

\subsection{Overview of results}

The aim of this article is to study resonances of $P(h)$ in the semiclassical limit $h \to 0$. Resonances are usually defined as poles of the meromorphically continued resolvent $(P(h)-z)^{-1}$. We define them via complex scaling, specifically as eigenvalues of a Fredholm operator obtained from $P(h)$ through an analytic dilation of the coordinates in the region exterior to the magnetic field (see Section~\ref{sec:def res} and Definition~\ref{def:resonances}). This approach allows us to transform the resonance problem into a spectral problem for a non-selfadjoint operator, whose discrete spectrum in a suitable complex sector corresponds precisely to the resonances of the original Hamiltonian.\medskip

In this work, we demonstrate that a local magnetic field of described type creates resonances in the semiclassical/strong field regime. We investigate five distinct scenarios and prove the following:

\begin{itemize}
    \item \emph{Constant fields:} If the magnetic field is locally constant, then there are resonances  whose real parts are close to the Landau levels, and whose imaginary parts are negative and exponentially small (see Theorem~\ref{thm:main} below).
    \item \emph{Isolated zeros of a local magnetic field:} If the magnetic field vanishes at a point, without loss of generality identified with the origin, and behaves locally like $r^\gamma$, where $r$ is a radial variable and $\gamma$ is a positive parameter, then the local harmonic potential of the previous case is replaced by an anharmonic one, the eigenvalues of which we call anharmonic Landau levels. For such fields, there are resonances whose real parts are close to these levels, and whose imaginary parts are negative and exponentially small (see Theorem~\ref{thm:main-ah} below).
\item \emph{Magnetic wells:} 
If the magnetic field attains a nondegenerate and positive local minimum (a magnetic
well), then it gives rise to resonances with a semiclassical expansion in powers of $h$ (see
Theorem~\ref{thm:main**} below).
\item \emph{Sharp magnetic interface:} If the magnetic
field is locally piecewise constant with a jump discontinuity across a curve modeling a sharp interface, 
and if the curvature of the discontinuity  has a local non-degenerate maximum, then there are resonances 
 which are asymptotically described by the field strength and the geometry of the interface (see Theorem~\ref{thm:main*} below).
\item \emph{Zero-field island:} If the magnetic field vanishes on an open set $\omega$ (a hole) and is positive in a neighborhood of $\partial\omega$, then there are resonances whose real parts are close to $h^2$ times the eigenvalues of the Dirichlet Laplacian on $\omega$ (see Theorem~\ref{thm:zero island} below).
\end{itemize}  

The asymptotic expansions of the real parts of these resonances exhibit distinct scaling regimes in the semiclassical parameter $h$, as summarized in Table~\ref{tab:resonance-expansions}. 

\begin{table}[htbp]
\centering
\caption{Asymptotics of the real part of resonances $z_n(h),$ $n=0,1,\dots,$ as $h\to0$.}
\label{tab:resonance-expansions}
\begin{tabular}{@{}p{0.38\textwidth} p{0.52\textwidth}@{}}
\toprule
\textbf{Local behavior} & \textbf{Real part asymptotics} \\
\midrule
Constant field & $\operatorname{Re} z_n(h)\sim (2n+1)h$ \\
Anharmonic Landau Hamiltonian& $\operatorname{Re} z_n(h)\sim \Lambda_n^\gamma h^{1+\frac{\gamma}{2+\gamma}}$ \\
Sharp magnetic interface & $\operatorname{Re} z_n(h)\sim a_0 h + a_1 h^{3/2} + (2n+1)a_2 h^{7/4}$ \\
Magnetic well & $\operatorname{Re} z_n(h)\sim b_0 h + (2n b_1+b_2) h^2$ \\
Zero-field island& $\operatorname{Re} z_n(h)\sim \ell_n h^2$ \\
\bottomrule
\end{tabular}
\end{table}

The real coefficients $\Lambda_n^\gamma$ and $\ell_n$ are eigenvalues of limiting operators, while the real coefficients $a_0,a_1,a_2$ and $b_0,b_1,b_2$,  depend on the geometric and spectral data of each configuration. Their explicit expressions, involving the flux, curvature, and magnetic well parameters, are provided in the statements of Theorems \ref{thm:main**} and \ref{thm:main*}, respectively.

\subsection{Statement of results}

\subsubsection*{Locally constant fields}

Our first result concerns locally constant magnetic fields\footnote{That is $B=b_0>0$ on an open set $\Omega_0$;  without loss of generality we may take $b_0=1$.}. Let $\mathbb{N}_0$ be the set of non-negative integers, and for all $n\in\mathbb N_0$, set $E_n(h)=(2n+1)h$.

\begin{theorem}[Existence of resonances near Landau levels]\label{thm:main}
Suppose that $\Curl\bA=1$ on a disk of radius $r_0\leq R_0$. For each  $n \in \mathbb{N}_0$ and $c\in(0,1/4)$, there exists $h_n > 0$ such that for all $h \in(0, h_n)$, the rectangle
\[
 [E_n(h) - h^{-2}e^{-cr_0^2/2h}, E_n(h) + h^{-2}e^{-cr_0^2/2h}] + i[-h^{-3}e^{-c r_0^2/h}, 0)
\]
contains at least one resonance of $P(h)$.
\end{theorem}

Generically, if for a fixed $h$ the resonances are simple (in the sense that the pole residue is a rank-one operator), we expect that their number will increase as $h\to 0$. In fact, along the proof of Theorem~\ref{thm:main}, we construct infinitely many linearly independent quasimodes with energy close to $E_n(h)$.
\medskip

The intuition behind Theorem~\ref{thm:main} is that for $\Curl\bA=1$ in a neighborhood of $0$, the operator $P(h)$ behaves locally like the Landau Hamiltonian, $P^{\rm L}(h)=(-ih\nabla -\bA^{\rm sym})^2$, where $\bA^{\rm sym}$ is expressed in the symmetric-gauge: $\bA^{\rm sym}=\frac12(-x_2,x_1)$. The spectrum of this operator consists of the Landau levels $E_n(h)=\Lambda_n h$, which are degenerate eigenvalues with infinite multiplicity.

\subsubsection*{Anharmonic Landau levels}
We consider now a  magnetic field $B^{\rm ah}$ with an isolated zero, defined by $B^{\rm ah}(x)=|x|^\gamma$ with $\gamma> 0$.  
Note that  a corresponding vector potential is
\[\bA^{\rm ah}(x)=\frac{|x|^\gamma}{2+\gamma}(-x_2,x_1)\qquad (x=(x_1,x_2)\in\R^2).\]
Since $B^{\rm ah}$ grows at infinity, the spectrum of the magnetic Laplacian 
$P^{\rm ah}=(-i\nabla-\bA^{\rm ah})^2$ is purely discrete (see e.g. \cite[Corollary~1.2]{HM}) and is given by an increasing sequence of eigenvalues 
\[\Lambda_0^\gamma<\Lambda_1^\gamma<\Lambda_2^\gamma<\dots.\]
We call $P^{\rm ah}$ the anharmonic Landau Hamiltonian, and we call $\Lambda_n^\gamma,$ $n\in\mathbb N_0$, the anharmonic Landau levels.
To justify this terminology, we express $P^{\rm ah}$ in polar coordinates $(r,\varphi)$,
\[P^{\rm ah}=-\partial_r^2-\frac{1}{r}\partial_r+\Bigl(-\frac{i}{r}\partial_\theta-\frac{r^{1+\gamma}}{2+\gamma}\Bigr)^2\]
and we note that the potential term in the radial variable gives rise to anharmonic potential $\propto r^{2+2\gamma}$ at $\infty$, in contrast to the Landau Hamiltonian case ($\gamma=0$). \medskip

We set
\[E_n^\gamma(h)=h^{1+\frac{\gamma}{2+\gamma}}\Lambda_n^\gamma.\]

\begin{theorem}[Resonances near anharmonic Landau levels] \label{thm:main-ah}
Let $\gamma>0$ and $p\in B(0,R_0)$. Suppose that $\Curl\bA(x)=B^{\rm ah}(x-p)$ on a neighborhood of $p$. For a given $n\in\mathbb N_0$, there exist positive constants $c,h_n$ such that
    the rectangle
\[
 [E_n^\gamma(h) - h^{-2}e^{-c/2h}, E_n^\gamma(h) +h^{-2}e^{-c/2h}] + i[-h^{-3}e^{-c/h}, 0)
\]
contains at least one resonance of $P(h)$, provided that $h\in(0,h_n]$.
\end{theorem}

Formally, for $\gamma=0$, the result in Theorem~\ref{thm:main-ah} is consistent with that of Theorem~\ref{thm:main}.

\subsubsection*{Local magnetic wells}

We turn now to the existence of  resonances created by a magnetic well. Suppose that $B^{\rm well}$ is a smooth  positive function on $\R^2$ having a unique non-degenerate minimum at a point $p_0\in B(0,R_0)$, that is
\begin{equation}\label{eq:m-well}
 B^{\rm well}(p_0)=b_0:=\inf_{x\in\R^2}B^{\rm well}(x)>0,\quad H=\frac12\mathrm{Hess}\,B^{\rm well}\big|_{p_0}\mbox{ is positive-definite}.   
\end{equation}
For all $n\in\mathbb N_0,$ let
\[e_n(h)=b_0h+\Bigl(2n\frac{\sqrt{\mathrm{det}\,H}}{b_0}+\frac{(\mathrm{Tr}\,H^{\frac12})^2}{b_0}\Bigr)h^2.\]

\begin{theorem}[Resonances near a magnetic well]\label{thm:main**}
Suppose that $\Curl\bA=B^{\rm well}$ on a neighborhood of $p_0$. Let $\alpha\in(0,1/2)$. For every $n\in\mathbb N_0$, there exist positive constants $L_n,h_n$ such that
    the rectangle
\[
 [e_n(h) - L_nh^{3}, e_n(h) + L_nh^{3}] + i[-h^{-3}e^{-1/h^{\alpha}}, 0)
\]
contains at least one resonance of $P(h)$, provided that $h\in(0,h_n]$.
\end{theorem}

Theorem~\ref{thm:main**} is local in the sense that only near $p_0$ the field equals $B^{\rm well}$. If the global field has several disjoint non-degenerate positive wells $p_j$, each individually produces resonances as in the theorem. When wells are well separated and symmetric (e.g. double magnetic wells), one expects exponential splitting $e^{-S/h}$ (with $S>0$ realted to some Agmon distance) due to tunneling (see \cite{EM, FGMR, FMR}). Proving this is delicate because the resonance window width $L_n h^3$ is polynomially large compared to the exponentially small splitting. Due to the competition between the resonance window and the exponential splitting,  a refined analysis is needed to resolve distinct resonances from different wells.

\subsubsection*{Locally step magnetic fields}

Our next result establishes the existence of resonances when the magnetic field is piecewise constant and sign-changing, with a jump discontinuity along a curve with non-constant curvature. One motivation for considering such a field comes from \cite{RP, RPM}. Before stating the result, we describe our geometric configuration (see Figure~\ref{fig:geometry_config}). Suppose that $\Gamma$ is
a simple smooth curve that splits $\R^2$ into two disjoint unbounded open sets, $P_1$ and $P_2$, and such
that i) $\Gamma$ is a semistraight line at infinity, and ii) $\Gamma$  intersects the circle $\partial B(0,R_0)$ transversely at two distinct points. Consider a point $x_0\in \Gamma\cap B(0,R_0)$ and a local arclength parametrization $s\in[-\epsilon_0,\epsilon_0]\mapsto M(s)$ of $\Gamma\cap \overline{(B(0,R_0)}$ such that $M(0)=x_0$. For a given $s\in[-\epsilon_0,\epsilon_0],$ consider a direct frame $(\mathbf t(s),\mathbf n(s))$ at $M(s)$ such that
\begin{itemize}
    \item $\mathbf n(s)$ is the unit normal of $\Gamma$ at the point $M(s)$ pointing towards $P_1$,
    \item $\mathbf t(s)$ is the unit tangent vector of $\Gamma$ at the point $M(s)$
\end{itemize}
We define the curvature $k$ of $\Gamma$ as $\mathbf t'(s) = k(s)\mathbf n(s)$ and we suppose that 
\[k(0)=k_0:=\max_{s\in[-\epsilon_0,\epsilon_0]}k(s)\in\R,\quad k''(0)=k_2<0,\quad k(s)<k(0) \mbox{ for }0<|s|\leq \epsilon_0,\]
hence the curvature of $\Gamma_0:=\Gamma\cap \overline{B(0,R_0)}$ has a non-degenerate and unique maximum at $x_0$.  Finally, we define the step function
\begin{equation}\label{eq:m-step}
B^{\rm step}_a(x)=\begin{cases}
    1&\mbox{if }x\in P_1,\\
    a\in(-1,0)&\mbox{if }x\in P_2.
\end{cases}
\end{equation}

We also need the spectral constants from \cite[(1.8)--(1.12)]{AHK}:
\begin{equation}\label{eq:spect-cst}
\beta_a=\mu_a(\zeta_a),\quad C_1(a)=\frac{1}{3}\left(1-\frac{1}{a}\right)\zeta_a\phi_a(0)\phi_a'(0) >0,\quad C_2(a)=\frac12\sqrt{\mu_a''(\zeta_a)C_1(a)}>0,
\end{equation}
where $\mu_a(\xi)$ is the lowest eigevalue of the Schr\"odinger operator 
\[h_a[\xi]=-\frac{d^2}{d\tau^2} + (\xi + b_a(\tau)\tau)^2 \quad \text{on } L^2(\mathbb{R}), \quad  b_a(\tau) = \mathbf{1}_{\mathbb{R}_+}(\tau) + a\,\mathbf{1}_{\mathbb{R}_-}(\tau),\]
$\zeta_a<0$ is the unique minimum of $\mu_a(\xi)$, $\xi\in\R$, and \(\phi_a\) is the positive eigenfunction of \(h_a[\zeta_a]\) of unit \(L^2\)-norm corresponding to \(\beta_a\).

\begin{theorem}[Curvature induced magnetic resonances]\label{thm:main*}
    Suppose that $\Curl\bA=B^{\rm step}_a$ on a neighborhood of $x_0$. Let the constants $\beta_a,C_1(a),C_2(a)$ be as in \eqref{eq:spect-cst}. There exists a positive constant $c$, 
    and for every  $n\in\mathbb N_0$, there exist positive constants $L_n,h_n$ such that, upon setting
    \[\lambda_n(h)=\beta_ah-k_0C_1(a)h^{3/2}+(2n+1)\sqrt{|k_2|}\,C_2(a)h^{7/4},\]
    the rectangle
\[
 [\lambda_n(h) - L_nh^{2}, \lambda_n(h) + L_nh^{2}] + i[-h^{-3}e^{-c/h^{1/8}}, 0)
\]
contains at least one resonance of $P(h)$, provided that $h\in(0,h_n]$.
\end{theorem}

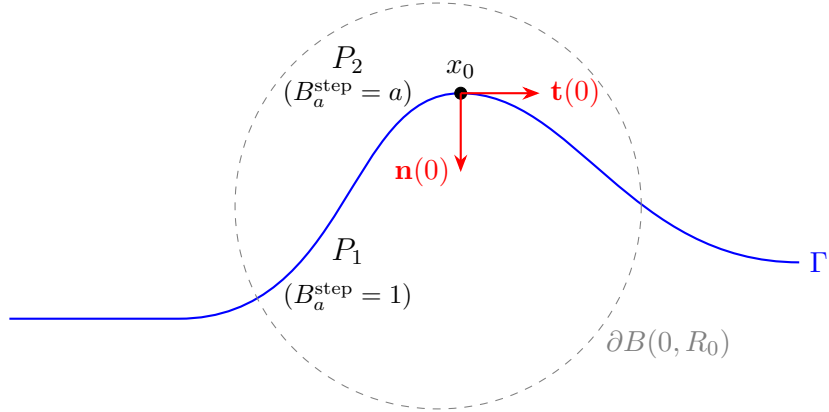
\begin{figure}[htbp]
    \centering
    \begin{tikzpicture}[scale=1.5]
        \draw[thick, blue] 
            (-3, -0.5) -- (-1.5, -0.5) 
            .. controls (0, -0.5) and (0, 1.5) .. (1, 1.5) 
            .. controls (2, 1.5) and (2.5, 0) .. (4, 0)
            node[right] {$\Gamma$};
        
        \coordinate (X0) at (1, 1.5);
        \filldraw (X0) circle (1.5pt) node[above=2pt] {$x_0$};

        \draw[-Stealth, thick, red] (X0) -- ++(0.7, 0) node[right] {$\mathbf{t}(0)$};
        \draw[-Stealth, thick, red] (X0) -- ++(0, -0.7) node[left] {$\mathbf{n}(0)$};

        \node[font=\large] at (0, 0.1) {$P_1$};
        \node[font=\small] at (0, -0.3) {($B_a^{\rm step}=1$)};
        
        \node[font=\large] at (0, 1.8) {$P_2$};
        \node[font=\small] at (0, 1.5) {($B_a^{\rm step}=a$)};

        \draw[dashed, gray] (0.8, 0.5) circle (1.8);
        \node[gray, below right] at (2.2, -0.5) {$\partial B(0, R_0)$};

    \end{tikzpicture}
    \caption{Illustration of the curve $\Gamma$ splitting $\mathbb{R}^2$ into regions $P_1$ and $P_2$ with jump discontinuity in the magnetic field.}
    \label{fig:geometry_config}
\end{figure}

We focus on the case when the field is sign changing leading to a strong localization along the edge; a classical particle in this situation follows a snake-shaped orbit \cite{RPM}. Recall that a step which is not sign changing, $a > 0$ in (1.6), also exhibits edge states propagating along the interface (as well as other Iwatsuka-type field variations, expected to be valid generally \cite[Section~6.5]{CFKS87}), however, in the Ginzburg-Landau counterpart of the current problem it gives rise to a ‘non-trapping magnetic step’ (see \cite[Remark 1.5]{AKS}). The semiclassical behavior of the low-lying eigenvalues can be then determined to leading order in light of \cite[Section~3.2]{AKS} (see also \cite[Proposition 4.2]{A}). However, the splitting of the eigenvalues remains open, along with the concentration/decay of the corresponding eigenfunctions. This is behind the technical issue that prevented us from transplanting these eigenvalues to resonances under such a local field.

\subsubsection*{Zero-field island} 
Our final result concerns the situation when the magnetic field vanishes in a subset of $B(0,R_0)$---a \emph{zero-field island}. Let us suppose that $\omega\subset B(0,R_0)$ is an open simply connected set with smooth boundary which is relatively compact in $B(0,R_0)$.  

We denote by 
\[\ell_0<\ell_1<\ell_2<\dots\]
the eigenvalues of the Dirichlet Laplacian in $\omega$.

We will prove that emergence of semi-classical resonances  near the eigenvalues $\ell_n$ when the magnetic field vanishes in $\omega$ while being positive in a neighborhood $\omega$. For the sake of simplicity, we will focus on a simple example of such a magnetic field: Suppose that $B^{\rm isl}:\mathbb \R^2\to\mathbb R$ satisfies $B^{\rm isl}(x)=0$ in $\omega$ and $B^{\rm isl}(x)=1$ in $\R^2\setminus\overline{\omega}$. Our result is:

\begin{theorem}[Resonances in a zero-field island]\label{thm:zero island}
    Suppose that  $\Curl \bA(x)=B^{\rm isl}(x)$ on $B(0,R_0)$. For every $n\in\mathbb N_0$, there exist positive constants $c,h_n$ and a function $\varepsilon_n(h)=o(1)$ such that the rectangle
\[
 [\ell_nh^2 - h^2\varepsilon_n(h) , \ell_nh^2 + h^2\varepsilon_n(h)] + i[-h^{-3}e^{-c/h^{1/2}}, 0)
\]
contains at least one resonance of $P(h)$, provided that $h\in(0,h_n]$.
\end{theorem}

We know from \cite{HH} that in the limit of infinite intensity the magnetic field produces a Dirichlet barrier; for large but finite field,  the particle can tunnel through giving rise to resonances, as confirmed by Theorem~\ref{thm:zero island}. This phenomenon is the magnetic analogue of the `well in an island' studied in \cite[Th\'eor\`eme~9.7]{HS} for the Schr\"odinger operator with a scalar potential. In the setting of \cite{HS}, resonances are exponentially close to the Dirichlet eigenvalues of the island, just as in Theorem~\ref{thm:zero island}.

\subsection{Examples}

We conclude the introduction by  giving examples of vector potentials satisfying our hypotheses.

\begin{example}[Radial step magnetic field]\label{example1}
The magnetic field \(
B(x) = \begin{cases}
1, & |x| < R_0 \\
0, & |x| > R_0
\end{cases}
\)
 is generated by the vector potential
\[
\bA(x) = \begin{cases}
\dfrac{(-x_2, x_1)}{R_0^2}, & |x| \le R_0, \\[1em]
\bA_\alpha(x), & |x| \ge R_0,
\end{cases}
\]
with flux $\alpha=R_0^2/2$. The conclusion of Theorem~\ref{thm:main} holds with $r_0=R_0$.
\end{example}

\begin{example}[Not necessarily radial magnetic field]\label{example2}
Consider  $B_0\in L^2(\mathbb R^2)$ with support $K_0\subset B(0,R_0)$, and let $\mathbf a_0\in H^1(B(0,R_0))$ be a vector potential satisfying\footnote{One may take $\mathbf a_0=(-\partial_2\phi_0,\partial_1\phi_0)$ where $\phi_0\in H^2(B(0,R_0))\cap H^1_0(B(0,R_0))$ solves $\Delta\phi_0=B_0$.} $\Curl\mathbf a_0=B_0$. Thanks to the smooth Urysohn's lemma, there is a cut-off function $\chi\in C_c^\infty(\mathbb R^2;[0,1])$ such that  $\chi(x)=1$ if $x\in K_0\cup B(0,\epsilon)$, for some $0<\epsilon<R_0,$  and $\chi(x)=0$ if $|x|\geq R_0$. Define the global vector potential:
\[\bA(x)=\chi(x)\mathbf a_0(x)+(1-\chi(x))\bA_{\alpha}(x),\]
where $\bA_\alpha$ is the Aharonov-Bohm potential with flux $\alpha=\frac1{2\pi}\int_{B(0,R_0)}B_0(x)dx$. Clearly, $\bA$ satisfies our hypotheses and the resulting field $B=\Curl\bA$ satisfies $B(x)=B_0(x)$ for $x\in K_0$. Consequently,
\begin{itemize}
    \item If $B_0=1$ locally, then by Theorem~\ref{thm:main}, there are resonances near the Landau levels, and the conclusion of Theorem~\ref{thm:main} holds with $r_0$ the inner radius of the set $\{B_0=1\}$,
    \item If $B_0=B^{\rm ah}(x-p)$ locally near a zero $p\in B(0,R_0)$, then by Theorem~\ref{thm:main-ah}, there are resonances near the energy levels  of the anharmonic Landau Hamiltonian,
    \item If $B_0=B_a^{\rm step}$ locally , then by Theorem~\ref{thm:main*}, there are curvature induced resonances near the energy levels $\lambda_n(h)$, $n\in\mathbb N_0$,
    \item If $B_0=B^{\rm well}$ locally  near the well with the positive minimum at $p_0$, then by Theorem~\ref{thm:main**}, there are resonances near the energy levels $e_n(h)$.
    \end{itemize}
\end{example}

\begin{example}[Zero-field island]\label{example3}Let $0<r_1<R_0$ and $\omega=B(0,r_1)$. For $r>0$, define $a(r)$ as follows: $a(r)=0$ for $r<r_1$, $a(r)=\frac{r^2-r_1^2}{2r}$ for $r_1\le r\le R_0$, and $a(r)=\frac{R_0^2-r_1^2}{2r}$ for $r\ge R_0$. Let $\mathbf A(x)=a(|x|)\frac{(-x_2,x_1)}{|x|}$.  Then $\mathbf A$ is $\mathbf A=\mathbf A_\alpha$ for $|x|\ge R_0$ with $\alpha=\frac{R_0^2-r_1^2}2$, and $
\curl
\mathbf A=B^{\rm isl} $ on $B(0,R_0)$ as in the hypothesis of Theorem~\ref{thm:zero island}.
\end{example}

\subsection{Organization}
The rest of the paper is organized as follows. In Section~\ref{sec:def res}, we verify that our operator $P(h)$ satisfies the black box assumptions in \cite{TZ} and define resonances as eigenvalues of a scaled operator following \cite{S}. Sections~\ref{sec:loc-cst}, \ref{sec:ah}, \ref{sec:step-f}, \ref{sec:m-well} and \ref{sec:zero island} are devoted to the proofs of Theorems~\ref{thm:main}, \ref{thm:main-ah}, \ref{thm:main*}, \ref{thm:main**} and \ref{thm:zero island}, respectively. These proofs share the common strategy of  constructing quasimodes and applying the resonance existence theorem of Tang and Zworski \cite{TZ}.

\section{Defining  resonances}\label{sec:def res}

We  present the rigorous definition of resonances for $P(h)$ within the black box scattering framework (see \cite{S, TZ}). After verifying the necessary assumptions, we introduce complex scaling to access the resonances as eigenvalues of a deformed operator.

\subsection{Verification of black box assumptions}\label{sec:blackbox}

Let $\mathcal{H} = L^2(\R^2)$ with orthogonal decomposition
\[
\mathcal{H} = \mathcal{H}_{R_0} \oplus L^2(\R^2 \setminus B(0,R_0)),
\]
where $\mathcal{H}_{R_0} = L^2(B(0,R_0))$ and  $R_0$ is the radius of the disk where $B$ is supported.  For simplicity of notation, we  write $ H^2(\R^2 \setminus B(0,R_0))$ instead of $ H^2(\R^2 \setminus \overline{B(0,R_0)})$.\medskip  

The aim of this section is to verify that our operator \(P(h)\) satisfies the hypotheses of the approach \cite[Section 2, p. 263]{TZ}, which we reproduce in Assumptions 2.1–2.5 below.

\begin{assumption}[Domain condition]\label{ass:D-C}
The projection map $\mathbb{1}_{\R^2\setminus B(0,R_0)}:\mathcal D\to H^2(\mathbb R^2\setminus B(0,R_0))$ is   bounded uniformly with respect to $h$ and has  a uniformly bounded right inverse.
\end{assumption}

The said boundedness refers to the graph norm  $\|(i+P(h))u\|_{L^2(\R^2)}$ in $\mathcal D$, and the semiclassical norm $\|u-h^2\Delta u\|_{L^2}$ in $H^2(\R^2\setminus B(0,R_0))$.

\begin{proof}[Verification of Assumption~\ref{ass:D-C}]
Thanks to \eqref{eq:Domain-H}, restricting a function in $\mathcal D$  to $\R^2\setminus B(0,R_0)$ produces a function in $H^2(\R^2 \setminus B(0,R_0))$, and conversely, a function in $H^2(\R^2 \setminus B(0,R_0))$ can be extended to a function in $H^2(\R^2)$, the domain of $P(h)$. Thus, $\mathbb{1}_{\R^2 \setminus B(0,R_0)} \mathcal{D} = H^2(\R^2 \setminus B(0,R_0))$ and $\mathbb{1}_{\R^2\setminus B(0,R_0)}$
has a right inverse. We have uniform control in $h$  because $\bA$ and its derivatives are bounded on the exterior of $B(0,R_0)$. 
\end{proof}

\begin{assumption}[Compactness]\label{Ass:C}
$\mathbb{1}_{B(0,R_0)} (P(h) + i)^{-1} : \mathcal{H} \to \mathcal{H}_{R_0}$ is compact.
\end{assumption}

\begin{proof}[Verification of Assumption~\ref{Ass:C}]
$(P(h)+i)^{-1}$ maps into $\mathcal{D}$, and its restriction to $B(0,R_0)$ gives $H^2(B(0,R_0))$ by \eqref{eq:Domain-H}.  The compact embedding $H^2(B(0,R_0)) \hookrightarrow L^2(B(0,R_0))$ yields the result.
\end{proof}

\begin{assumption}[Exterior differential operator]\label{ass:ext-op}
For $u \in \mathcal{D}$,
\[
\mathbb{1}_{\R^2 \setminus B(0,R_0)} P(h) u = Q(h)(u|_{\R^2 \setminus B(0,R_0)}),
\]
where $Q(h)$ is a formally self-adjoint differential operator on $L^2(\R^2)$ of the form\footnote{We denote $D_{x_k}:=i\partial_{x_k}$ and  $D_x^{\alpha}:=D_{x_1}^{\alpha_1}D_{x_2}^{\alpha_2}$ for $\alpha=(\alpha_1,\alpha_2)\in\mathbb N_0^2$.}
\[
Q(h)v = \sum_{|\alpha| \leq 2} a_\alpha(x;h)(hD_x)^\alpha v,
\]
with coefficients satisfying:
\begin{itemize}
    \item $a_\alpha(x;h) = a_\alpha(x)$ independent of $h$ for $|\alpha| = 2$,
    \item $a_\alpha(x;h) \in C_b^\infty(\R^2)$ uniformly bounded in $h$,
    \item $\sum_{|\alpha|=2} a_\alpha(x;h)\xi^\alpha \geq (1/c)|\xi|^2$ for some $c > 0$,
    \item $\sum_{|\alpha| \leq 2} a_\alpha(x;h)\xi^\alpha \to \xi^2$ uniformly in $h$ as $|x| \to \infty$.
\end{itemize}
\end{assumption}
\begin{proof}[Verification of Assumption~\ref{ass:ext-op}]
We define $Q(h)$ as a regularized Aharonov--Bohm Hamiltonian near the orgin. Choose a cutoff function $\eta \in C_c^\infty(\R^2)$ with:
\begin{itemize}
\item $\eta(x) = 1$ for $|x| \le R_0/2$,
\item $\eta(x) = 0$ for $|x| \ge R_0$,
\item $0 \le \eta \le 1$.
\end{itemize}
Define a smoothed Aharonov–Bohm potential:
\begin{equation}\label{eq:reg AB}
\tilde{\bA}(x) =  (1 - \eta(x))\left(-\frac{\alpha x_2}{|x|^2},\frac{ \alpha x_1}{|x|^2}\right).
\end{equation}
Then set
\[
Q(h) = (-ih\nabla - \tilde{\bA})^2.\]
Since
$\tilde{\bA}$ is smooth on all $\R^2$ and for $|x| \ge R_0$,  $\tilde{\bA}(x) =\bA(x)$, $Q(h)$ exactly  matches $P(h)$ outside $B(0,R_0)$.
Moreover, for $|x| \le R_0/2$, $\tilde{\bA}(x) = 0$, so $Q(h) = -h^2\Delta$ near the origin.
All coefficients of $Q(h)$ are in $C_b^\infty(\R^2)$ and satisfy the ellipticity and asymptotic conditions.
\end{proof}

\begin{assumption}[Analytic continuation]\label{Ass:AC}
There exist $\theta_0 \in [0,\pi)$, $\epsilon > 0$, and $R \geq R_0$ such that the coefficients $a_\alpha(x;h)$ of $Q(h)$ extend holomorphically in $x$ to the region
\[
\mathcal U=\{r\omega : \omega \in \mathbb{C}^2, \operatorname{dist}(\omega, \mathbb S^1) < \epsilon, r \in \mathbb{C}, |r| > R, \arg r \in [-\epsilon, \theta_0 + \epsilon)\},
\]
and the convergence $\sum_{|\alpha| \leq 2} a_\alpha(x;h)\xi^\alpha \to \xi^2$ as $|x| \to \infty$ remains valid in this larger set.
\end{assumption}

\begin{proof}[Verification of Assumption~\ref{Ass:AC}]
In Assumption~\ref{Ass:AC}, $\mathbb S^1$ is the unit circle, $\{(x_1,x_2)\in\R^2\colon x_1^2+x_2^2=1\}$. Let $w = (w_1, w_2) \in \mathbb{C}^2$ be the complexification of the real coordinates $x = (x_1, x_2) \in \R^2$ such that $|x|\geq  R_0$. In the exterior of $B(0,R_0)$, $\bA=\bA_\alpha$, $\Div\bA=0$, and the coefficients of $Q(h)$ are given by:
\[\begin{gathered}
    a_{(2,0)}(w) = a_{(0,2)}(w) = 1, \quad a_{(1,1)}(w) = 0 \\
    a_{(1,0)}(w)  = \frac{2\alpha w_2}{q(w)},\quad 
    a_{(0,1)}(w)  = -\frac{ 2\alpha w_1}{q(w)} \\
    a_0(w) = \alpha^2\frac{w_1^2 + w_2^2}{(q(w))^2} = \frac{\alpha^2}{q(w)}
\end{gathered}\]
where $q(w)=w_1^2+w_2^2$. Note that  $q(w)$ does not vanish on $\mathcal{U}$ if we choose $R>R_0$ sufficiently large and $\epsilon>0$ sufficiently small. In fact, if $w = r\omega \in \mathcal{U}$, we have $q(w) = r^2 q(\omega)$ with $q(\omega)=1$ if $\omega\in \mathbb S^1$. By continuity, we get $|q(w)|\geq c>0$ if $w\in\mathcal U$. Since $q(w)$ is holomorphic, we get that the functions $a_{(1,0)}(w)$, $a_{(0,1)}(w)$, and $a_0(w)$ are holomorphic in $\mathcal U$. Moreover, the sum 
\[\sum_{|\alpha|\leq 2}a_\alpha(w;h)\xi^\alpha=\xi^2+\frac{2\alpha w_2h}{q(w)}\xi_1-\frac{2\alpha w_1 h}{q(w)}\xi_2+\frac{\alpha^2}{q(w)}\to \xi^2\]
as $|w|\to+\infty$ uniformly with respect to $h$.\medskip
\end{proof}

The last assumption  in the blackbox scattering theory concerns the compactified operator $P^{\#}$ in 
\[\mathcal H^{\#}=\mathcal H_{R_0}\oplus L^2(M\setminus B(0,R_0)), \]
where $M$ is the torus $(\R/\tilde R\mathbb{Z})^2$, with $\tilde R> 4R$ and $R>R_0$ is from Assumption~\ref{Ass:AC}. The disk $B(0,R_0)$ can be viewed as a subset of $M$. The operator $P^{\#}$ is equal to $P(h)$ in $B(0,R_0)$. More precisely, we take (see \cite[Remark 7.1]{SZ})
\[ P^{\#}=P(h)\chi +Q^\# (1-\chi),\quad Q^\#=(-ih\nabla-\zeta\tilde{\mathbf A})^2\]
where $\chi\in C_c^{\infty}(B(0,2R))$ is equal to $1$ on $B(0,R)$,  $\zeta\in C_c^\infty\bigl((-\tilde R/2,\tilde R/2)^2\bigr)$ is equal to $1$ on $B(0,2R)$, $\tilde{\mathbf A}$ is the vector potential in \eqref{eq:reg AB}, and  the domain of $P^{\#}$ is 
\[\mathcal D^{\#}=\{u\in\mathcal H^{\#}\colon \chi u\in \mathcal D,(1-\chi)u\in H^2(M)\}.\]
Since $\zeta\tilde{\mathbf A}$ is a smooth $\tilde R\mathbb Z^2$-periodic vector field, $Q^\#$ is a self-adjoint elliptic operator on $L^2(M)$ with domain $H^2(M)$, and $Q^\#=P(h)$ on $\operatorname{supp}\nabla\chi\subset\{R\le|x|\le2R\}$. Thus,  $P^\#$ is self-adjoint, as in \cite{SZ}.

\begin{assumption}[Eigenvalue counting for compactified operator]\label{Ass:Weyl}
For $\lambda \geq 1$, it holds
\[
N(P^{\#}(h), [-\lambda, \lambda]) = \mathcal{O}(\lambda/h^2).
\]
\end{assumption}

\begin{proof}[Verification of Assumption~\ref{Ass:Weyl}]
$P^{\#}(h)$ is an elliptic self-adjoint operator on a compact 2-dimensional manifold, so its eigenvalue counting function satisfies a Weyl law $N(\lambda) \sim C\lambda/h^2$, thereby verifying Assumption~\ref{Ass:Weyl}.\medskip
\end{proof}

All assumptions \ref{ass:D-C}--\ref{Ass:Weyl} are satisfied, so the black box scattering theory applies to $P(h)$.

\subsection{Complex scaling}\label{sec:complex scaling}

Following the framework of Sj\"ostrand \cite[Section 5]{S}, we introduce complex scaling to define resonances for the magnetic Hamiltonian $P(h)$. 

\subsubsection{The scaled operator}

Let $\epsilon_0 > 0$ be small and $R_1 > R_0$ be chosen sufficiently large. Define a smooth family of functions $f_\theta: [0,+\infty) \to \mathbb{C}$, depending on a parameter $\theta \in [0,\theta_0]$ with $\theta_0 < \pi/2$, satisfying:

\begin{itemize}
    \item[(i)] $f_\theta(t) = t$ for $0 \le t \le R_1$ (no deformation near the black box interior),
    \item[(ii)] $0 \le \arg f_\theta(t) \le \theta$ and $\partial_t f_\theta(t) \neq 0$,
    \item[(iii)] $\arg f_\theta(t) \le \arg \partial_t f_\theta(t) \le \arg f_\theta(t) + \epsilon_0$,
    \item[(iv)] $f_\theta(t) = e^{i\theta} t$ for $t \ge T_0$, where $T_0$ depends only on $\epsilon_0$ and $R_1$.
\end{itemize}

In polar coordinates $x = t\omega$ with $t = |x|$ and $\omega \in \mathbb S^1$, we define the map
\[
\kappa_\theta: \mathbb{R}^2 \ni x = t\omega \mapsto f_\theta(t)\omega \in \mathbb{C}^2.
\]

The image $\Gamma_\theta = \kappa_\theta(\mathbb{R}^2)$ is a maximally totally real (m.t.r.) submanifold of $\mathbb{C}^2$. By construction, $\Gamma_\theta$ coincides with $\mathbb{R}^2$ on $B(0,R_1)$ and is rotated by angle $\theta$ at infinity (see Figure~\ref{fig:gamma_theta} for illustration).

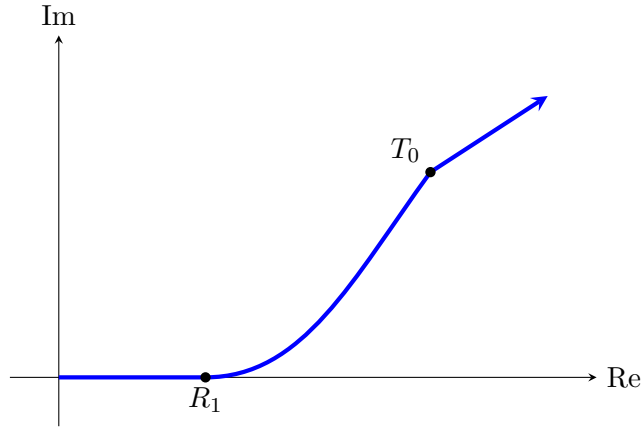
\begin{figure}[htbp]
    \centering
    \begin{tikzpicture}[>=stealth, scale=1.3]
        \draw[->] (-0.5,0) -- (5.5,0) node[right] {$\text{Re}$};
        \draw[->] (0,-0.5) -- (0,3.5) node[above] {$\text{Im}$};
        \coordinate (O) at (0,0);
        \coordinate (R1) at (1.5,0);
        \coordinate (T0_pt) at (3.8, 2.1); 
        \draw[ultra thick, blue] (O) -- (R1);
        \draw[ultra thick, blue] (R1) .. controls (2.5,0) and (3,1) .. (T0_pt);
        \draw[ultra thick, blue, ->] (T0_pt) -- (5, 2.88);
        \fill (R1) circle (1.5pt) node[below] {$R_1$};
        \fill (T0_pt) circle (1.5pt) node[above left] {$T_0$};
    \end{tikzpicture}
    \caption{Illustration of $\Gamma_\theta$: The profile $f_\theta(t)$ remains real for $t \le R_1$ and rotates to $e^{i\theta}$ for $t \ge T_0$.}
    \label{fig:gamma_theta}
\end{figure}

We identify $\Gamma_\theta$ with $\mathbb{R}^2$ via the parametrization $\kappa_\theta$. Define the Hilbert space
\[
\mathcal{H}_\theta = \mathcal{H}_{R_0} \oplus L^2(\Gamma_\theta \setminus B(0,R_0)).
\]

Choose a cutoff function $\chi \in C_c^\infty(B(0,R_1))$ with $\chi = 1$ near $\overline{B(0,R_0)}$. For $u \in \mathcal{H}_\theta$, we define the scaled operator $P_\theta(h)$ by
\[
P_\theta(h)u = P(h)(\chi u) + Q_\theta(h)((1-\chi)u),
\]
where $Q_\theta(h)$ is the differential operator obtained by substituting $x = \kappa_\theta(y)$ into $Q(h)$ and expressing it in the coordinates $y \in \mathbb{R}^2$.
More precisely,
\[
Q_\theta(h) = (-ih\nabla_y - \tilde{\bA}_\theta(y))^2,
\]
where $\tilde{\bA}_\theta(y) = \tilde{\bA}(\kappa_\theta(y)) \cdot d\kappa_\theta(y)$ is the pullback of the vector potential $\tilde\bA$ in \eqref{eq:reg AB}. In coordinates, if $\kappa_\theta(y) = (\kappa_\theta^1(y), \kappa_\theta^2(y))$, then
\[(\tilde{\bA}_\theta)_j(y) = \sum_{i=1}^{2} \tilde{\bA}_i(\kappa_\theta(y)) \frac{\partial \kappa_\theta^i}{\partial y_j}(y).
\]

\subsubsection{Resonances}

Thanks to Assumptions~\ref{ass:D-C}--\ref{Ass:Weyl} verified in Section~\ref{sec:blackbox}, we can apply the general theory of complex scaling developed in \cite{SZ} and \cite[Section 5]{S}, particularly \cite[Lemma 5.1]{S} and the discussion following \cite[(5.8)]{S}. We obtain that for each $\theta \in [0,\theta_0]$ with $\theta_0 < \pi/2$, the scaled operator $P_\theta(h)$ is elliptic and that $P_\theta(h)-z$ is Fredholm of index $0$ for $z \in \mathbb{C} \setminus e^{-2i\theta}[0,\infty)$.  Consequently, by Analytic Fredholm Theory, the spectrum of $P_\theta(h)$ in $\mathbb{C} \setminus e^{-2i\theta}[0,\infty)$ is discrete\footnote{For $\theta=0$, this spectrum, if non-empty, is contained in $(-\infty,0)$.}.

A key property established in \cite[Lemma 5.2]{S} is the independence of the spectrum on the choice of scaling angle $\theta$, provided we stay away from the rotated real axis. Specifically\footnote{We denote $e^{-2i[\theta_1,\theta_2]}[0,\infty):=\{ e^{-2i\theta}t : \theta \in [\theta_1, \theta_2], t \in [0, \infty) \}$.}, for $0 \leq \theta_1 < \theta_2 \le \theta_0$ and $z_0 \in \mathbb{C} \setminus e^{-2i[\theta_1,\theta_2]}[0,\infty)$, the dimensions of $\operatorname{Ker}(P_{\theta_1}(h) - z_0)$ and $\operatorname{Ker}(P_{\theta_2}(h) - z_0)$ coincide. 
This independence allows us to define resonances unambiguously:

\begin{definition}[Resonances]\label{def:resonances}
A point $z_0 \in e^{-2i[0,\theta_0)}(0,\infty)$ is called a \textbf{resonance} of $P(h)$ if and only if $z_0 \in \sigma(P_\theta(h))$ for some  $\theta \in (0,\theta_0)$ with $z_0 \in e^{-2i[0,\theta)}(0,\infty)$.
\end{definition}

For a resonance $z_0$, we define its \textbf{multiplicity} as the rank of the spectral projection
\[
\pi_{\theta,z_0} = \frac{1}{2\pi i} \int_{\mathcal C} (z - P_\theta(h))^{-1} dz,
\]
where $\mathcal C$ is a small positively oriented circle around $z_0$ (see \cite[(5.12)]{S}). This rank is independent of the choice of $\theta$ in the sense that, if $z_0\in e^{-2i[0,\tilde\theta)}(0,\infty)\cap\sigma(P_{\tilde\theta}(h))$ with $\tilde\theta\in[0,\theta_0)$, then $\pi_{\tilde\theta,z_0}=\pi_{\theta,z_0}$.

\begin{remark}
    For the equivalence between complex scaling and meromorphic continuation definitions of resonances in the context of magnetic Hamiltonians with singular vector potentials, we refer to \cite[Section 3]{Y}.
\end{remark}

\subsection{The Tang–Zworski resonance criterion}

We recall the following result from \cite[Theorem, p.~264]{TZ}, which will be our main tool for proving the existence of resonances. 
The statement involves three functions $S(h)$, $R(h)$ and $w(h)$ defined on a right neighborhood of $0$ and satisfying
\[
\begin{gathered}
    De^{-D/h} \leq S(h) \quad\text{for some constant } D>0,\\[2mm]
    R(h) = \mathcal O(h^\infty) \quad\text{and}\quad R(h) \ll S(h) \quad\text{as } h\to0^+,\\[2mm]
    \frac{w(h)^2}{h^{-3} S(h)} \to +\infty \quad\text{as } h\to0^+.
\end{gathered}
\]
A convenient choice satisfying the last condition is $w(h) = h^{-2} \sqrt{S(h)}$.

\begin{theorem}[Tang–Zworski \cite{TZ}]\label{thm:TZ}
Let $P(h)$ be a semiclassical black box operator satisfying Assumptions~\ref{ass:D-C}--\ref{Ass:Weyl}. 
Suppose there exist normalized functions $\mathfrak u_j(h)\in \mathcal D$, $j\in J(h)$, with $\operatorname{supp} \mathfrak u_j(h)\subset B(0,R_0)$ such that
\[
\|(P(h)-\mathcal E_j(h))\mathfrak u_j(h)\|_{L^2(\R^2)} = \mathcal O(R(h)),
\]
where $\mathcal E_j(h)\in [E_0-C_0h,\,E_0+C_0h]$ for some constants $E_0,C_0$, and $|\mathcal E_j(h)-\mathcal E_{j'}(h)|\geq w(h)$ for $j\neq j'$. 
Then for sufficiently small $h$, there exist resonances $z_i(h)$ of $P(h)$ in
\[
[\mathcal E_j(h)-w(h),\,\mathcal E_j(h)+w(h)] + i[-h^{-3}S(h),\,0].
\]
\end{theorem}

the proof relies on Lemma~1 in \cite{TZ}, which seems to require $E_0>0$ in order to obtain a constant uniform with respect to $h$. We provide in the appendix\footnote{See Theorem~\ref{rev:threshold}, and Lemmata~\ref{rev:AP} and \ref{rev:apriori}.} the supplementary estimates needed for the above theorem to continue to hold when the energies are positive but approach the threshold $E_0 = 0$ as $h\to 0$.


\section{Resonances near Landau levels: Proof of Theorem~\ref{thm:main}}\label{sec:loc-cst}

We now turn to the  proof of Theorem~\ref{thm:main} by constructing quasimodes supported near a region where $B=1$ and applying the black box resonance existence theorem of Tang and Zworski.

\subsection{Quasimodes}\label{sec:quasimodes}

Consider  the symmetric-gauge expression $\bA^{\rm sym}(x) = \frac{b}{2}(-x_2, x_1)$, $b>0$. The corresponding magnetic Laplacian  is the Landau Hamiltonian:
\[
H^{\text{L}} = \left(-i\partial_{x_1} + \frac{b}{2}x_2\right)^2 + \left(-i\partial_{x_2} - \frac{b}{2}x_1\right)^2,
\]
which can be expressed in polar coordinates $(r,\varphi)$ as
\[H^{\text{L}}=-\partial_r^2-\frac1r\partial r+\left(-\frac{i}{r}\partial_\varphi-\frac{br}{2}\right)^2.\]
Its eigenvalues are
\[
\Lambda_n(b) = (2n+1)b, \quad n = 0,1,2,\dots,
\]
with eigenfunctions
\[
\psi_{n,m}(r,\varphi) = C_{n,m} \, r^{|m|} e^{im\varphi} e^{-b r^2/4} L_n^{|m|}\left(\frac{b r^2}{2}\right), \quad m \in \mathbb{Z},
\]
where $L_n^{|m|}$ are associated Laguerre polynomials, and $C_{n,m} \sim b^{(|m|+1)/2}$ is a constant to ensure that $\psi_{n,m}$ is normalized in $L^2(\R^2)$. \medskip

Consider $r_0>0$. For any $\delta \in (0,1)$, let $\chi \in C_c^\infty(\mathbb{R}^2)$ be a smooth cutoff satisfying:
\begin{itemize}
    \item $\chi(x) = 1$ for $|x| \le (1-\delta)r_0$,
    \item $\chi(x) = 0$ for $|x| \ge r_0$,
    \item $0 \le \chi \le 1$, $|\nabla \chi| \le C/(\delta r_0)$, $|\Delta \chi| \le C/(\delta r_0)^2$.
\end{itemize}

Define the quasimode
\[
u_{n,m}(x) = \chi(x) \psi_{n,m}(x).
\]
Note that $u_{n,m}\in C_c^\infty(\R^2)$ with support in the disk $B(0,r_0)$.

\begin{lemma}\label{lem:quasimode}
For any fixed $n,m$ and $\delta \in (0,1)$, the quasimode $u_{n,m}$ satisfies
\[\begin{gathered}
    \|u_{n,m}\|_{L^2(\R^2)} = 1 + \mathcal O\big( b^{|m|+1}e^{- (1-\delta)^2 r_0^2b/2} \big) ,\\
    \|(H^{\rm L}-\Lambda_n(b))u_{n,m}\|_{L^2(\R^2)} = \mathcal O\big(b^{|m|+2} e^{- (1-\delta)^2 r_0^2b/4} \big).
\end{gathered}
\]
\end{lemma}

\begin{proof}
Note that the full-space Landau eigenfunctions $\psi_{n,m}$ are normalized,
$\|\psi_{n,m}\|_{L^2(\mathbb{R}^2)} = 1$.  
Since $\chi \equiv 1$ for $|x| \le (1-\delta)r_0$ and $0 \le \chi \le 1$ elsewhere, we have  
\[
\|u_{n,m}\|^2_{L^2(\R^2)} = \int_{\mathbb{R}^2} \chi^2 |\psi_{n,m}|^2 \, dx 
= 1 - \int_{|x| > (1-\delta)r_0} (1-\chi^2) |\psi_{n,m}|^2 \, dx.
\]  
The error term is bounded by the tail of the eigenfunction outside $|x| > (1-\delta)r_0$,
\[
1 - \|u_{n,m}\|^2_{L^2(\R^2)} \le \int_{|x| > (1-\delta)r_0} |\psi_{n,m}|^2 \, dx.
\]  
For large $r$, the asymptotic behavior of $|\psi_{n,m}|^2$ is dominated by the Gaussian factor $e^{-b r^2/2}$, with a polynomial prefactor $r^{2|m|+4n}$.  
Thus the radial tail integral is proportional to  
\[
\int_{(1-\delta)r_0}^\infty r^{2|m|+4n+1} e^{-b r^2/2} \, dr
= 2^{|m|+2n} b^{-(|m|+2n+1)} \int_{t_\delta}^\infty t^{|m|+2n} e^{-t} \, dt,
\]  
where $t_\delta =  (1-\delta)^2 r_0^2b / 2$.  
For large $b$, the incomplete gamma function  in the above expression
 behaves as $\int_{t_\delta}^\infty t^p e^{-t} dt \sim t_\delta^p e^{-t_\delta}$, so the tail integral is $\mathcal{O}\big( e^{-(1-\delta)^2 r_0^2b/2} \big)$.  
Hence $\|u_{n,m}\|^2 = 1 + \mathcal{O}(C_{n,m}^2e^{- (1-\delta)^2 r_0^2b/2})$, which yields the first estimate in the lemma.\medskip

Using that $H^{\text{L}}\psi_{n,m} = \Lambda_n(b)\psi_{n,m}$ and $u_{n,m} = \chi \psi_{n,m}$, we obtain
\[
(H^{\text{L}}-\Lambda_n(b))u_{n,m} = [H^{\text{L}}, \chi] \psi_{n,m}.
\]
Since $\chi$ commutes with multiplication, the commutator expands to:
\[
[H^{\text{L}}, \chi] = -(\Delta \chi) - 2(\nabla \chi) \cdot \nabla - 2i\bA \cdot (\nabla \chi).
\]
Thus,
\[
(H^{\text{L}}-\Lambda_n(b))u_{n,m}(x) = f(x) \psi_{n,m}(x) + g(x) \cdot \nabla \psi_{n,m}(x),
\]
with $f,g$ supported in the annulus $(1-\delta)r_0 \le |x| \le r_0$ and bounded by constants depending on $\delta$ (but independent of $n,m$ and $b$). Consequently, we have
\begin{multline*}
\|(H^{\text{L}}-\Lambda_n(b))u_{n,m}\|_{L^2(\R^2)}^2 \le C_\delta\int_{(1-\delta)r_0 \le |x| \le r_0} (|\psi_{n,m}|^2 + |\nabla \psi_{n,m}|^2) \, dx\\
=\mathcal O(C_{n,m}^2b^2e^{-(1-\delta)^2r_0^2b/2}),
\end{multline*}
thereby establishing the second estimate in the lemma.
\end{proof}

\subsection{From quasimodes to resonances}\label{sec:resonances}

Having verified the black box scattering assumptions in Section~\ref{sec:blackbox} and constructed quasimodes in Section~\ref{sec:quasimodes}, we now apply the theorem of Tang and Zworski \cite{TZ} to deduce the existence of resonances exponentially close to the real axis.

\begin{proof}[Proof of Theorem~\ref{thm:main}]
Suppose that $B=1$ on a disk $B(p,r_0)\subset B(0,R_0)$. Modulo a gauge transformation $\bA\to\bA-\nabla\phi$ in $B(p,r_0)$, the new vector potential  is given by the symmetric-gauge expression, $\bA_{\rm sym}(x-p)$. 

Consider $\alpha\in(0,1/4)$ and the quasimodes $u_n:=u_{n,0}$ from Lemma~\ref{lem:quasimode} (we choose $\delta$ sufficiently small so that $\alpha<(1-\delta)^2/4$). We modify $u_n$ as follows
\begin{enumerate}[i)]
\item Translation: $u_n(x)\mapsto u_n(x-p)$, to get a support in $B(p,r_0)$,
\item  Multiplication by a pure phase: $u_n \mapsto e^{i\phi/h}u_n$, to account for the gauge transformation $\bA(x)\to \bA_{\rm sym}(x-p)=\bA(x)-\nabla \phi(x)$.
\end{enumerate}
Let $\tilde u_n$ denote the modified quasimode, that is
\[\tilde u_n(x)=\begin{cases}
    e^{i\phi(x)/h}u_n(x-p)&\mbox{if }|x-p|<r_0,\\
    0&\mbox{if }|x-p|\geq r_0.
\end{cases}\]
In terms of the semiclassical parameter
$h = 1/b$, we have  for each $n \in \mathbb{N}_0$, 
\begin{itemize}
    \item $\|\tilde u_n\|_{L^2(\R^2)} = 1 + \mathcal{O}(e^{-\alpha r_0^2/h})$,
    \item $\operatorname{supp} \tilde u_n \subset B(p,r_0)$,
    \item $\|(P(h) - E_n(h))\tilde u_n\|_{\mathcal{H}} = \mathcal{O}(e^{-\alpha r_0^2/h})$,
\end{itemize}
where 
\[
E_n(h) = (2n+1)h\in (E_0-C_0h,E_0+C_0h)\quad\mbox{with }E_0=0\mbox{ and }C_0=2n+2.
\]

For distinct indices $n \neq m$, we have
\[
|E_n(h) - E_m(h)| = 2|n-m|h \geq 2h.
\]
We now apply Theorem~\ref{thm:TZ} of Tang and Zworski with $\mathcal E_n(h)=E_n(h)$ and
\begin{itemize}
    \item The error bound $R(h) = C e^{-\alpha r_0^2/h}$,
    \item $S(h) = e^{-c r_0^2/h}$ where $0<c<\alpha$,
    \item $w(h)=h^{-2}e^{-cr_0^2/2h}$.
\end{itemize}
We conclude that there exists at least one resonance of $P(h)$ in
\[
 [E_n(h) - w(h), E_n(h) + w(h)] + i[-h^{-3}S(h), 0].
\]
Since $[0,+\infty)$ is the spectrum of $P(h)$, a resonance cannot be real. 
\end{proof}

\section{Resonances near anharmonic Landau levels: Proof of Theorem~\ref{thm:main-ah}}\label{sec:ah}

The proof  is similar to that of Theorem~\ref{thm:main} and relies on the construction of quasi-modes with exponentially small remainder. These quasi-modes are related to the eigenfunctions of the the anharmonic Landau Hamiltonian. Thus, we begin by presenting some spectral properties of this hamiltonian.

\subsection{The anharmonic Landau Hamiltonian}

Let $\gamma>0$ and $\bA^{\rm ah}=\frac{|x|^{\gamma}}{2+\gamma}(-x_2,x_1)$. Consider the operator $P^{\rm ah}(b)=(-i\nabla-b\bA^{\rm ah})^2$, with $b>0$. Setting its domain to be
\[\{u\in L^2(\R^2)\colon (-i\nabla-b\bA^{\rm ah})u\in L^2(\R^2;\mathbb C^2),(-i\nabla-b\bA^{\rm ah})^2u\in L^2(\R^2)\},\]
the operator $P^{\rm ah}(b)$ is self-adjoint operator in $L^2(\R^2)$, thanks to Friedrichs
 theorem.

Due to the scaling $x\mapsto b^{\frac{1}{2+\gamma}}x$, we reduce to case $b=1$ and obtain the spectrum of $P^{\rm ah}(b)$ as
\[\sigma(P^{\rm ah}(b))=b^{\frac{2}{2+\gamma}}\sigma(P^{\rm ah}).\]

Moreover, if $\psi(x;b)$ is a normalized eigenfunction of $P^{\rm ah}(b)$, then it can be expressed via the normalized eigenfunction of $P^{\rm ah}$ as follows
\[\psi(x;b)=b^{\frac{1}{2+\gamma}}\psi(b^{\frac{1}{2+\gamma}}x ;1).\]

We next prove that the operator $P^{\rm ah}$ has a compact  resolvent.

\begin{proposition}\label{prop:comp-res}
Let $R=(P^{\rm ah}+1)^{-1}$. The operator $R:L^2(\R^2)\to L^2(\R^2)$ is compact.
\end{proposition}
\begin{proof}
Note that $B^{\rm ah}=\Curl\bA^{\rm ah}=|x|^\gamma$. The space
\[X=\{u\in L^2(\R^2)\colon |u|\in H^1(\R^2),~|x|^{\gamma/2}u\in L^2(\R^2)\}, \]
with the norm $\|u\|_X=\|\nabla|u|\|_{L^2(\R^2)}+\|u\|_{L^2(\R^2)}+\||x|^{\gamma/2}u\|_{L^2(\R^2)}$, is compactly embedded in $L^2(\R^2)$. Using the inequality
\[\int_{\R^2}|(-i\nabla-\bA)u|^2dx\geq \int_{\R^2}B^{\rm ah}|u|^2\,dx\quad (u\in C_c^\infty(\R^2)), \]
and the diamagnetic inequality $|\,\nabla|u|\,|\leq |(-i\nabla-\bA^{\rm ah})u|$, we get that $R:L^2(\R^2)\to X$ is continuous.
\end{proof}

We denote the sequence of strictly increasing eigenvalues of $P^{\rm ah}$ by $(\Lambda_n^\gamma)_{n\geq 0}$, and by the scaling law, the $n$'th eigenvalue of $P^{\rm ah}(b)$ as
\[\Lambda_n^\gamma(b)=b^{\frac{2}{2+\gamma}}\Lambda_n^\gamma.\]

We now turn to the decay of eigenfunctions, which we express in polar coordinates $(r,\varphi)$.

\begin{proposition}\label{label:dec-ef-ah}
Let $n\in\mathbb N_0$.  There exists an integer $m$ and a radial function $f$ such that, for all $b>0,$ the function
\[\psi_n(x;b)=b^{\frac{1}{2+\gamma}} e^{i m\varphi}f(b^{\frac{1}{2+\gamma}}r),\]
is a normalized eigenfunction of $P^{\rm ah}(b)$ with eigenvalue $\Lambda_n^\gamma(b)$. Furthermore, there exist positive constant $c_0,M$ such that 
\[\int_{\R_+}\Bigl(|f'(r)|^2+ |f(r)|^2\Bigr)e^{2c_0r^{2+\gamma} }rdr\leq M.\]
\end{proposition}
\begin{proof}
    By exploiting the rotational symmetry, the operator $P^{\rm ah}(b)$ is unitarily equivalent to the direct sum of the operators
    \[\bigoplus_{m\in\mathbb Z}\mathfrak h_m(b) \quad\mbox{ in }L^2(\R_+,rdr)\otimes L^2(\mathbb S^1)\cong L^2(\R^2), \]
    where
    \[\mathfrak h_m(b)=-\frac{d^2}{dr^2}-\frac{1}{r}\frac{d}{dr}+\Bigl(\frac{m}{r}-\frac{b r^{1+\gamma}}{2+\gamma} \Bigr)^2.\]
    By the scaling $r\mapsto b^{\frac{1}{2+\gamma}} r$, we reduce the task to the case $b=1$. Choose $m\in\mathbb Z$ such that $\Lambda_n^\gamma$ is an eigenvalue of $\mathfrak h_m(1)$. Then, a corresponding eigenfunction of $P^{\rm ah}$ is
    \[\psi_n(x;1)=e^{im\varphi}f(r)\]
    where $f$ is an eigenfunction of $\mathfrak h_m(1)$, normalized in $L^2(\R_+,rdr)$. By the scaling law, we obtain the eigenfunction $\psi_n(x;b)$ of $P^{\rm ah}(b)$.

    We establish the decay of $f$ at infinity by using the method of Agmon. Let $\chi\in C_c^\infty(\R_+)$. Starting from the identity,
    \[\mathfrak h_m(1)f=\Lambda_n^\gamma f,\]
    and taking the inner product with $e^{2c_0\chi}f$ in $L^2(\R_+,rdr)$, we get by integration by parts,
    \[\int_{\R_+}\bigl(f'(r)(e^{2c_0\chi}f(r))'+V_m (r)|e^{c_0\chi(r)}f(r)|^2\bigr)rdr=\Lambda_n^\gamma\int_{\R_+}|e^{c_0\chi(r)}f(r)|^2rdr,\]
    where
    \[V_m(r)=\Bigl(\frac{m}{r}-\frac{r^{1+\gamma}}{2+\gamma} \Bigr)^2.\]
    Noticing that
    \[f'(e^{2c_0\chi}f)'=[(e^{c_0\chi}f)']^2-c_0^2(\chi')^2(e^{c_0\chi}f)^2,\]
    we get eventually
    \begin{equation}\label{eq:IMS}
    \int_{\R_+}\Bigl(|(e^{c_0\chi(r)}f(r))'|^2+\bigl(V_m(r)-\Lambda_n^\gamma-c_0^2|\chi'(r)|^2\bigr)|e^{c_0\chi(r)}f(r)|^2\Bigr)rdr=0, 
    \end{equation}
By density, \eqref{eq:IMS} holds for $\chi\in H^1_0(\R_+)$. For a given integer $N\geq 2$, we choose $\chi=\chi_N$ where
\[\chi_N(r)=\begin{cases}
    r^{2+\gamma}&\mbox{if }0\leq r<N,\\
    (2+\gamma)N^{1+\gamma}(r_N-r)&\mbox{if }N\leq r\leq r_N:=\frac{3+\gamma}{2+\gamma}N,\\
    0&\mbox{if }r\geq r_N.
\end{cases}\]
Note that, for all $r\in\R_+$, we have $|\chi'_N(r)|\leq (2+\gamma)r^{1+\gamma}$ and 
\[V_m(r)\geq \frac{r^{2+2\gamma}}{(2+\gamma)^2}-\frac{2m}{2+\gamma}r^\gamma.\]
Choose $R_0> 1$ sufficiently large and $c_0$ sufficiently small such that
\[V_m(r)-\Lambda_n^\gamma-c_0^2(2+\gamma)^2r^{2+2\gamma}\geq 1,\quad\forall\,r\geq R_0.\]
Going back to \eqref{eq:IMS} and splitting the integral on the intervals $(0,R_0)$ and $(R_0,+\infty)$, we obtain
\[\begin{split}\int_{R_0}^{+\infty}\Bigl(|(e^{c_0\chi(r)}f(r))'|^2+\bigl(V_m(r)-\Lambda_n^\gamma-c_0^2|\chi'(r)|^2\bigr)|e^{c_0\chi(r)}f(r)|^2\Bigr)rdr
\\
\leq \int_{0}^{R_0}(\Lambda_n^\gamma+c_0^2|\chi'(r)|^2)|e^{c_0\chi(r)}f(r)|^2rdr.\end{split}\]
Thus, using the last inequality together with \eqref{eq:IMS} for $c_0=0$ and the fact that $V_m\ge 0$, we conclude that
 \[\int_{0}^N\Bigl(|(e^{c_0\chi_N(r)}f(r))'|^2+|e^{c_0\chi_N(r)}f(r)|^2\Bigr)rdr\leq M_0 e^{2c_0R_0^{2+\gamma}}\int_{0}^{R_0}|f(r)|^2rdr,\]
 with
 \[M_0=2\Lambda_n^\gamma+c_0^2\|\chi_N'\|_{L^\infty(0,R_0)}^2=2\Lambda_n^\gamma+(2+\gamma)^2c_0^2R_0^{2+2\gamma},\]
holds for all $N>R_0$; to finish the proof, we take $N\to+\infty$ and conclude by monotone convergence.
\end{proof}

\subsection{Quasi-modes and existence of resonances}

Following the construction in Section~\ref{sec:quasimodes}, we define quasimodes for the anharmonic Landau Hamiltonian $P^{\rm ah}(b)$. For a fixed $n \in \mathbb{N}_0$, let $\psi_n(x;b)$ be the normalized eigenfunction from Proposition~\ref{label:dec-ef-ah}.

Consider $r_0 > 0$ and a smooth cutoff $\chi \in C_c^\infty(\mathbb{R}^2)$ such that $\chi(x)=1$ for $|x| \le (1-\delta)r_0$ and $\chi(x)=0$ for $|x| \ge r_0$. We define the  quasimode:
\[
u_n(x;b) = \chi(x) \psi_n(x;b).
\]

\begin{lemma}\label{lem:quasimode-ah}
For fixed $n$ and $\delta \in (0,1)$, there exist positive constants $b_0,c_1$ such that, for all $b\geq b_0$, the quasimode $u_n$ satisfies:
\[\begin{gathered}
    \|u_n\|_{L^2(\R^2)} = 1 + \mathcal{O}\big( e^{-c_1 b} \big), \\
    \|(P^{\rm ah}(b) - \Lambda_n^\gamma(b))u_n\|_{L^2(\R^2)} = \mathcal{O}\big( e^{-c_1 b} \big).
\end{gathered}\]
\end{lemma}
\begin{proof}
The proof follows the same lines as Lemma~\ref{lem:quasimode}. The error terms are controlled by the tail of the eigenfunction $\psi_n$. By Proposition~\ref{label:dec-ef-ah}, the eigenfunction decays as $e^{-c_0 (b^{\frac{1}{2+\gamma}} r)^{2+\gamma}} = e^{-c_0 b r^{2+\gamma}}$. Integrating this tail over $|x| > (1-\delta)r_0$ yields the exponential smallness in $b$. The commutator $[P^{\rm ah}(b), \chi]$ is supported in the annulus where $\psi_n$ and its gradient are exponentially small, leading to the second estimate.
\end{proof}

\begin{proof}[Proof of Theorem~\ref{thm:main-ah}]
Assume $B = |x|^\gamma$ on a disk $B(p, r_0)$. By the gauge transformation $\bA \to \bA - \nabla \phi$, the operator $P(h)$ in $B(p, r_0)$ is unitarily equivalent to $P^{\rm ah}(b)$ shifted to $p$.

Let $h = 1/b$. We define the modified quasimodes $\tilde{u}_n(x) = e^{i\phi/h} u_n(x-p; 1/h)$. From Lemma~\ref{lem:quasimode-ah}, we have:
\begin{itemize}
    \item $\|\tilde u_n\|_{L^2} = 1 + \mathcal{O}(e^{-c_1/h})$,
    \item $\operatorname{supp} \tilde u_n \subset B(p, r_0)$,
    \item $\|(P(h) - E_n^\gamma(h))\tilde u_n\| = \mathcal{O}(e^{-c_1/h})$, where $E_n^\gamma(h) = h^{\frac{2+2\gamma}{2+\gamma}} \Lambda_n^\gamma$.
\end{itemize}

Since the eigenvalues $\Lambda_n^\gamma$ are strictly increasing, the separation between levels $E_n^\gamma(h)$ and $E_{n+1}^\gamma(h)$ is of order $h^{\frac{2+2\gamma}{2+\gamma}}$.

We now apply Theorem~\ref{thm:TZ} with $\mathcal E_n(h)=E_n^\gamma(h)$ and
\[R(h)=e^{-c_1/h},\quad S(h) = e^{-c/h},\quad  w(h) = h^{-2}e^{-c/2h}\quad\text{ for $0 < c < c_1$}.\]
The theorem guarantees the existence of a resonance
\[
z_n\in [E_n^\gamma(h) - w(h), E_n^\gamma(h) + w(h)] + i[-h^{-3}S(h), 0).
\]
\end{proof}

\section{Curvature induced resonances: Proof of Theorem~\ref{thm:main*}}\label{sec:step-f}

In this section, we suppose that $\Curl\bA=B_a^{\rm step}$ in a disk $U_0\subset B(0,R_0)$ centered at $x_0$ and of radius $\delta$, where $B_{a}^{\rm step}$ is the step magnetic field introduced in \eqref{eq:m-step}. Recall that $B_{a}^{\rm step}$ is a step function with a jump discontinuity across the curve $\Gamma$, and that $x_0$ is the point where $\Gamma$ has a non-degenerate maximal curvature.

Before heading to the proof of Theorem~\ref{thm:main*}, we collect a few results concerning semiclassical eigenvalue estimates under magnetic steps. Let $\Omega=B(0,R_0)$ and consider a vector potential $\bA^{\rm step}\in H^1(\Omega)$ such that $\Curl\bA^{\rm step}=B_a^{\rm step}$ in $\Omega$.  Let $P^{\rm step}_{\Omega,h}=(-ih\nabla-\bA^{\rm step})^2$ be the Dirichlet Laplacian in $L^2(\Omega)$ with the step magnetic field $B_a^{\rm step}$. This operator was studied in \cite{AHK}, where the following is proved:
\begin{itemize}
    \item The spectrum of $P^{\rm step}_{\Omega,h}$ is purely discrete and consists of the sequence of eigenvalues
    \[\tilde\lambda_0(h)\leq\tilde\lambda_1(h)\leq \dots\]
    repeated according to multiplicity,
    \item For every $n\in\mathbb N_0,$ the eigenvalue $\tilde\lambda_n(h)$ satisfies\footnote{The remainder term is shown to be $\mathcal O(h^{15/8})$ in \cite{AHK}. However, by WKB analysis \cite[Theorem~5.1]{FHK},  it can be improved to $\mathcal O(h^{2})$.} (see \cite[Theorem~1.2]{AHK})
    \begin{equation}\label{eq:asymp-ln}
\tilde\lambda_n(h)=\beta_ah-k_0C_1(a)h^{3/2}+(2n+1)\sqrt{|k_2|}\,C_2(a)h^{7/4}+\mathcal O(h^{2}),
    \end{equation}
    \item If $\psi_n$ is a normalized eigenfunction corresponding to the eigenvalue $\tilde\lambda_n(h)$, then it is localized near the point $x_0$ in the following sense (see \cite[Eq. (6.3) and (6.5)]{AHK})
    \begin{equation}\label{eq:dec-ef}
\int_{\Omega} \left( |\psi_{n}|^2 + h^{-1} |(h\nabla - i\bA^{\rm step})\psi_{n}|^2 \right) \exp\left( 2\alpha h^{-1/8} |x-x_0| \right) dx \le C,
    \end{equation}
    for some constants $\alpha,C>0$. Moreover, the constant $\alpha$ is independent of $n$.
\end{itemize}

\begin{proof}[Proof of Theorem~\ref{thm:main*}]
The proof relies on quasimodes construction and the theorem by Tang--Zworski \cite{TZ}. Let 
$\chi\in C_c^\infty(\R^2;[0,1])$ be a cut-off function such that
\begin{itemize}
    \item $\chi(x)=0$ in the exterior of the disk $U_0=B(x_0,\delta)$,
    \item $\chi(x)=1$ in the disk $\{|x-x_0|\leq \delta/2\}$.
\end{itemize}
For all $n\in\mathbb N_0,$ consider the quasimode
\[u_n(x)=c_n\chi(x)\psi_n(x)e^{i\phi(x)/h},\]
where $\phi:U_0\to \mathbb R$ satisfies
\[\bA-\nabla\phi=\bA^{\rm step}\quad\mbox{in }U_0,\]
and $c_n$ is a constant to enusre that $u_n$ is normalized in $L^2(\R^2)$. By \eqref{eq:dec-ef}, the constant $c_n$ approaches $1$ asymptotically,
\[c_n=1+\mathcal O(e^{-2\alpha_0/ h^{1/8}}),\]
for some $\alpha_0>0$.

Note that
\[P(h)u_n=P^{\rm step}_{\Omega,h}\,\tilde u_n \mbox{ in }\R^2,\]
where $\tilde u_n(x)=\chi(x)\psi_n(x).$ Thus,
\[\bigl\|\bigl(P(h)-\tilde\lambda_n(h)\bigr)u_n\bigr\|_{L^2(\R^2)}\leq \bigl\|[P_{\Omega,h}^{\rm step},\chi]\psi_n\bigr\|_{L^2(\R^2)}=\mathcal O(e^{-2\alpha_1/h^{1/8}})\]
by \eqref{eq:dec-ef}, where $\alpha_1>0$ is a constant.

Now we apply Theorem~\ref{thm:TZ} with 
\[\mathcal E_n(h)=\tilde \lambda_n(h),\quad R(h)=e^{-\alpha_0/h^{1/8}},\quad S(h)=e^{-c/h^{1/8}},\quad w(h)=h^{-2}e^{-c/2h^{1/8}}\,,\]
where $0<c<\alpha_1$. We obtain that the rectangle
\[[\tilde\lambda_n(h) - w(h), \tilde\lambda_n(h) + w(h)] + i[-h^{-3}e^{-c/h^{1/8}}, 0]
\]
contains at least one resonance of $P(h)$.
To conclude, we use   $\tilde\lambda_n(h)=\lambda_n(h)+\mathcal O(h^{2})$ by \eqref{eq:asymp-ln}.
\end{proof}

\section{Resonances induced by magnetic wells: Proof of Theorem~\ref{thm:main**}}\label{sec:m-well}

In this section, we turn to the proof of Theorem~\ref{thm:main**}, which relies on semiclassical asymptotics for magnetic wells \cite{HK}. Suppose that $\Curl\bA=B^{\rm well}$ in a disk $U_0\subset B(0,R_0)$ centered at $p_0$ and of radius $\delta$, where $B^{\rm well}$ is a smooth function on $\R^2$ having a non-degenerate positive minimum at $p_0$ as in \eqref{eq:m-well}. Consider the magnetic Laplacian  $P^{\rm well}_{h}=(-ih\nabla-\bA^{\rm well})^2$ in $L^2(\R^2)$, where $\bA^{\rm well}$ is any vector potential satisfying $\Curl\bA^{\rm well}=B^{\rm well}$. We recall the following facts regarding the spectrum of $P^{\rm well}$:
\begin{itemize}
    \item For a given non-negative integer $N$, there exists a $h_0>0$ such that for $0<h<h_0$, the bottom of the spectrum of $P^{\rm well}_{h}$ is a discrete eigenvalue, and $P^{\rm well}_{h}$ has at least $N$ discrete eigenvalues below the essential spectrum (if any) ordered as 
    \[\tilde e_0(h)\leq\tilde e_1(h)\leq \dots \leq \tilde e_N(h)\]
    and repeated according to multiplicity,
    \item For every $n\in\mathbb N_0,$ the eigenvalue $\tilde e_n(h)$ satisfies (see \cite[Theorem~1.2]{HK} and \cite[Corollary~1.7]{RV})
    \begin{equation}\label{eq:asymp-en}
\tilde e_n(h)=b_0h+\Bigl(2n\frac{\sqrt{\mathrm{det}\,H}}{b_0}+\frac{(\mathrm{Tr}\,H^{\frac12})^2}{b_0}\Bigr)h^2+\mathcal O(h^{3}),
    \end{equation}
    \item If $\psi_n$ is a normalized eigenfunction corresponding to the eigenvalue $\tilde e_n(h)$, then it is localized near the well $p_0$ in the following sense (see \cite[Proposition~4.1]{RV}): For a given $(\alpha,r)\in(0,1/2)\times\R_+$, there is a constant $C$ such that\footnote{If we impose analyticity assumptions on $B^{\rm well}$, we get exponential decay of order $e^{-\epsilon/h}$, see \cite{GRV}.}
    \begin{equation}\label{eq:dec-ef-en}
\int_{|x-p_0|\geq r} \left( |\psi_{n}|^2 + h^{-1} |(h\nabla - i\bA^{\rm well})\psi_{n}|^2 \right) dx \le Ce^{-1/h^{\alpha}}.
    \end{equation}
\end{itemize}

\begin{proof}[Proof of Theorem~\ref{thm:main**}]
The proof relies on quasimodes construction. Let 
$\chi\in C_c^\infty(\R^2;[0,1])$ be a cut-off function such that
\begin{itemize}
    \item $\chi(x)=0$ in the exterior of the disk $U_0=B(p_0,\delta)$,
    \item $\chi(x)=1$ in the disk $\{|x-p_0|\leq \delta/2\}$.
\end{itemize}
For all $n\in\mathbb N_0,$ consider the quasimode
\[u_n(x)=c_n\chi(x)\psi_n(x)e^{i\phi(x)/h},\]
where $\phi:U_0\to \mathbb R$ satisfies
\[\bA-\nabla\phi=\bA^{\rm well}\quad\mbox{in }U_0,\]
and $c_n$ is a constant (asymptotic to $1$) to $u_n$ is normalized in $L^2(\R^2)$.

We have
\[P(h)u_n=P^{\rm well}_{h}\,\tilde u_n \mbox{ in }\R^2,\]
where $\tilde u_n(x)=\chi(x)\psi_n(x).$ 
Thus, by \eqref{eq:dec-ef-en},
\[\bigl\|\bigl(P(h)-\tilde e_n(h)\bigr)u_n\bigr\|_{L^2(\R^2)}=\mathcal O(e^{-1/h^{\alpha}}).\]

Applying Theorem~\ref{thm:TZ}, we obtain that for a given $\alpha\in(0,1/2)$, the rectangle
\[[\tilde e_n(h) - h^{-2}e^{-1/2h^\alpha}, \tilde e_n(h) + h^{-2}e^{-1/2h^\alpha}] + i[-h^{-3}e^{-1/h^{\alpha}}, 0]
\]
contains at least one resonance of $P(h)$. In light of \eqref{eq:asymp-en}, this finishes the proof of Theorem~\ref{thm:main**}.
\end{proof}

\section{Zero-field island: Proof of Theorem~\ref{thm:zero island}}\label{sec:zero island}

This section is devoted to the proof of Theorem~\ref{thm:zero island}. We suppose that $\Curl\bA=B^{\rm isl}$ in a relatively open subset of $\tilde\omega$, where $B^{\rm inst}$ is the step function equal to $0$ in $\omega$ and to $1$ in the the exterior of $\omega$. The proof of Theorem~\ref{thm:zero island} relies on the construction of quasi-modes defined by the eigenfunctions of a specific auxiliary operator: a magnetic Laplacian in a bounded set, with a vector potential $\bA^{\rm isl}$ satisfying
\[\Curl\bA^{\rm isl}=B^{\rm isl}\quad\mbox{in }\R^2.\]

\subsection{An auxiliary problem}
Recall that $\omega\subset\overline\omega\subset\tilde\omega$ and that $\partial\omega$ is smooth. We can find an open subset $\hat\omega$ with smooth boundary such that $\omega\subset\hat\omega\subset\tilde\omega$. 

For all $b>0$, consider the magnetic Laplacian $\mathcal L(b):=(-i\nabla-b\bA^{\rm isl})^2$ in $L^2(\hat\omega)$, with magnetic Neumann boundary condition. Let $(\hat\ell_n(b))_{n\geq 0}$ be the strictly increasing sequence of eigenvalues of $\mathcal L(b)$.

For a given fixed $n\in\mathbb N_0$, the eigenvalue $\hat\ell_n(b)$ converges to the $n$'th eigenvalue $\ell_n$ of the Dirichlet Laplacian $-\Delta$ in $L^2(\omega)$. For a proof see \cite{GKS} (and also \cite{KS}).

We would like to understand the behavior of the eigenfunctions, notably their decay in the set $\hat\omega\setminus\omega$, as $b\to+\infty$. We prove a quantitative  decay estimate in the next proposition, which involves the distance function
\[t_\omega(x)=\begin{cases}
0&\mbox{if }x\in \omega,\\
\mathrm{dist}(x,\omega)&\mbox{if }x\not\in\omega.    
\end{cases}\]

\begin{proposition}[Decay of eigenfunctions]\label{prop:dec-ef-zero island}
    Let $n\in\mathbb N_0$. There exists positive constant $b_0$ and $M$ such that, for all $b\geq b_0$, if $\psi_n$ is a normalized eigenfunction of $\mathcal L(b)$ with eigenvalue $\hat\ell_n(b)$, then
    \[\int_{\hat\omega\setminus\omega} \Bigl(|\psi_n|^2+|(-i\nabla-b\bA^{\rm isl})\psi_n|^2\Bigr)\exp\bigl(\mbox{$\frac12$} b^{1/2}t_\omega(x)\bigr)dx\leq M.\]
\end{proposition}
\begin{proof}
Let $\Phi(x)=\epsilon_0 b^{1/2}t_\omega(x)$. We will fix a choice of $\epsilon_0>0$ later on. Using the eigenvalue equation $(-i\nabla-b\bA^{\rm isl})^2\psi_n=\hat\ell_n(b)\psi_n$ and integration by parts, we obtain
\begin{equation}\label{eq:Agmon-zero island}\int_{\hat\omega}\Bigl(|(-i\nabla-b\bA^{\rm isl})(e^{\Phi}\psi_n)|^2- \bigl(\epsilon_0^2b|\nabla t_\omega|^2+\hat\ell_n(b) \bigr)|e^{\Phi}\psi_n|^2\Bigr)dx=0.
\end{equation}
Since the magnetic field is constant in $\tilde\omega\setminus\omega$, the lowest eigenvalue of $(-i\nabla-b\bA^{\rm isl})^2$ in $\hat\omega\setminus\omega$ behaves asymptotically like $\Theta_0b$ (see \cite[Theorem~8.11]{FH-b}), where $\Theta_0>1/2$ is a constant (called the de\,Gennes constant). Thus, we have by the min-max principle
\[\frac{b}{2}\int_{\hat\omega\setminus\omega}|e^{\Phi}\psi_n|^2dx\leq \int_{\hat\omega\setminus\omega}|(-i\nabla-b\bA^{\rm isl})(e^{\Phi}\psi_n)|^2dx. \]
In the set $\omega,$ $\Phi=0$ and we control the kinetic energy by the eigenvalue $\hat\ell_n(b)$ as follows
\[\int_{\omega}\Bigl(|(-i\nabla-b\bA^{\rm isl})(e^{\Phi}\psi_n)|^2\leq \int_{\hat\omega}\Bigl(|(-i\nabla-b\bA^{\rm isl})\psi_n|^2\leq \hat\ell_n(b).\]
Returning back to \eqref{eq:Agmon-zero island} and using that $|\nabla t_\omega(x)|=\mathbf 1_{\omega}(x)$ a.e., we write after splitting the integral over $\omega$ and $\hat\omega\setminus\omega$,
\[\int_{\hat\omega\setminus\omega}\Bigl(\frac12|(-i\nabla-b\bA^{\rm isl})(e^{\Phi}\psi_n)|^2+\bigl(\frac{b}{4}-\epsilon_0^2b-\hat\ell_n(b) \bigr)|e^{\Phi}\psi_n|^2\Bigr)dx\leq \hat\ell_n(b).\]
To conclude, we choose $\epsilon_0=1/4$ and $b_0$ sufficiently large so that $\frac14b-\hat\ell_n(b)\geq \frac18 b$.
\end{proof}
\subsection{Quasimodes}

We now move to the construction of quasimodes using the eigenfunctions $\psi_n$, $n\in\mathbb N_0$, of the auxiliary operator $\mathcal L(b)$. Choose $\chi\in C_c^\infty(\R^2;[0,1])$  such that $\chi=1$ in $\overline\omega$ and $\mathrm{supp\,}\chi\subset\hat\omega$. Let $h=1/b$ and
\[u_n(x)=\chi(x)\psi_n(x)\times (\mbox{gauge transformation}).\]
Then, by Proposition~\ref{prop:dec-ef-zero island},
\[
\begin{gathered}
    \|u_n\|_{L^2(\R^2)}=1+\mathcal O(e^{-c/h^{1/2})},\\
    \|(P(h)-h^2\hat\ell_n(1/h))u_n\|_{L^2(\R^2)}=\mathcal O(e^{-c/h^{1/2}}),
\end{gathered}
\]
for some constant $c>0$. 

\begin{proof}[Proof of Theorem~\ref{thm:zero island}]
Applying Theorem~\ref{thm:TZ}  with the quasimodes $(u_n)$ and the energy levels $(h^2\hat\ell_n(1/h))$, we obtain resonances
\[z_n\in [h^2\hat\ell_n(1/h)-h^{-2}e^{-c/2h^{1/2}}]+i[ -h^{-3}e^{-c/h^{1/2}},0].\]
To conclude, we use that $\hat\ell_n(1/h)\to \ell_n$ as $h\to0^+$.
\end{proof}

\subsection*{Acknowledgments} The authors would like to thank Joachim Asch,  William Borrelli and Davide Fermi for fruitful discussions. \medskip

\section*{Declarations and statements}~\medskip

{\bf Research funding}. A.K. was partially supported by a startup fund at AUB (grant no. 513125) and by URB grant (award no. 104752).\\

{\bf Conflict of Interest}. The authors declare that they have no competing interests regarding the publication of this paper.\\

{\bf Availability of data}. There are no data associated with the research in this paper.\\

{\bf Use of AI tools}. The authors acknowledge the exploratory use of AI tools  in preparing the appendix. The proofs are independently written and checked by the authors.\\


\appendix


\section{Tang-Zworski Theorem under an a priori bound}
 
Fix $\theta\in(0,\pi/8)$  and write $\Sigma_\theta:=\mathbb C\setminus e^{-2i\theta}[0,\infty)$. Recall that resonances do not depend on $\theta$, and 
by Definition 2.6, $\Res P(h)\subset e^{-2i[0,\theta)}(0,\infty)\subset\Sigma_\theta$. For $\chi\in C_c^\infty(\mathbb R^2)$
with $\chi=1$ near $B(0,R_0)$ and $\Ima z>0$, $\chi(P(h)-z)^{-1}\chi=\chi(P_\theta(h)-z)^{-1}\chi$ \cite[Lemma~3.5]{SZ}, and the
right-hand side extends holomorphically to $\Sigma_\theta\setminus\Res P(h)$; we denote the extension by the same symbol.\medskip

We note first that there are no real resonances.

\begin{lemma}[No real resonances]\label{rev:noreal}
$\Res P(h)\cap(0,\infty)=\emptyset$.
\end{lemma}

\begin{proof}
A resonance $z_0>0$ would be an eigenvalue of $P(h)$ embedded in $[0,\infty)$ \cite[\S3]{SZ}. On $\{|x|>R_0\}$ an eigenfunction $u$
solves $((-ih\nabla-\mathbf A_\alpha)^2-z_0)u=0$ and lies in $L^2$. Decomposing $u=\sum_{m\in\mathbb Z}e^{im\varphi}u_m(r),$ each mode $u_m(r)$ satisfies a Bessel equation of order
$|m-\alpha/h|$ with spectral parameter $z_0/h^2>0$, whose only solution in $L^2((R_0,+\infty),rdr)$  is $u_m=0$. Hence $u=0$ outside
$B(0,R_0)$, and $u\equiv0$ in $\mathbb R^2$ by unique continuation.
\end{proof}

We state a scaled version of  the Tang-Zworski theorem \cite{TZ}.
 
\begin{theorem}[Quasimodes to resonances at scale $s(h)$]\label{rev:threshold}
Let $P(h)$ satisfy Assumptions 2.1--2.5. Let $s:(0,h_0]\to(0,1]$, $[a,b]\subset(0,\infty)$, and let $\widetilde\Omega_0\subset\Sigma_\theta$
be a compact simply connected neighbourhood of $[a,b]$. Assume
\begin{itemize}
\item[\textbf{(AP)}] there is $\mathcal N:(0,h_0]\to[1,\infty)$ such that for $h\in(0,h_0]$ and $g\in(0,\frac12]$,
\[
\big\|\chi(P(h)-z)^{-1}\chi\big\|\le\exp\Big(\mathcal N(h)\log\frac2g\Big)\quad
\text{for }z\in s(h)\widetilde\Omega_0\setminus\bigcup_{z_j\in\Res P(h)}D(z_j,s(h)g).
\]
\end{itemize}
Suppose for $h\in(0,h_0]$ there are $E(h)\in s(h)[a,b]$ and $u(h)\in\mathcal D$ with $\|u\|=1$, $\supp u\subset B(0,R_0)$,
$\|(P(h)-E(h))u(h)\|\le R(h)$. Let $S,D,w:(0,h_0]\to(0,1)$ satisfy
\begin{align}
&S\le D\le w,\qquad w/D\to\infty,\qquad D\le s(h),\qquad [E\pm6w]+i[-2D,S]\subset s(h)\widetilde\Omega_0, \label{rev:C1}\\
&w^2\ \ge\ C_1\,D^2\,\mathcal N(h)\log\frac{4s(h)}{D}, \label{rev:C2}\\
&R(h)\ \le\ \frac{S}{4e}\exp\Big(-\mathcal N(h)\frac SD\log\frac{4s(h)}D\Big), \label{rev:C3}
\end{align}
with $C_1$ a large absolute constant. Then for $h$ small, $([E-6w,E+6w]+i[-2D,0))\cap\Res P(h)\neq\emptyset$.
\end{theorem}
\begin{remark}
    Theorem~\ref{thm:TZ} remains valid when the energies $\mathcal E_j(h)$ lie in $h^p[a,b]$ with $[a,b]\subset(0,\infty)$: for $1\le p<2$ without further assumption, and for $p=2$ under the additional hypothesis (H) below.
\end{remark}
\begin{proof}[Proof of Theorem~\ref{rev:threshold}]
 By Lemma~\ref{rev:noreal}, a resonance in $[E-6w,E+6w]+i[-2D,0]$ has negative imaginary part, which gives the half-open interval in the statement.
The proof relies on a general complex-analysis lemma from Tang-Zworski, \cite[Lemma~2]{TZ}. We state it here in a form adapted to our setting\footnote{It is obtained from the Tang-Zworski statement by renaming the constants: the Gaussian width
$\alpha=Sh^{-n^\sharp-1}$ becomes $D$, the half-width $h^k$ becomes $w$, and the a priori bound $Ae^{Ah^{-n^\sharp}\log(1/hS)}$ becomes $M$. The bound $|F|\le1/\operatorname{Im}z$ is used in \cite{TZ} only on the top edge, and the factor $e^{S\sigma}=(MS)^{S/(D+S)}$, which \cite{TZ} bound by a constant using $S\ge De^{-D/h}$, is kept explicit here.}.
For $E\in\mathbb R$, $0<S\le D\le w$,  $M>0$, and $\Omega:=[E-5w,E+5w]+i[-D,S]$, if $F$ is holomorphic on a neighborhood of $\Omega$ with
$|F|\le M$ on $\Omega$ and $|F|\le1/S$ on the top edge $[E-5w,E+5w]+iS$, then there are absolute constants $\kappa_0,C_0$ such that if $w/D\ge\kappa_0$ and
$w^2/D^2\ge C_0\log(C_0MS)$, we have 
\begin{equation}\label{eq:TZ-Lem2.2}
\sup_{[E-w,E+w]}|F|\le2e\,S^{-1}(MS)^{S/(D+S)}.
\end{equation}
The rest of the proof follows verbatim that of \cite[\S4]{TZ}. Suppose $[E\pm6w]+i[-2D,0]$  contains no resonance of $P(h)$ for a sequence $h=h_j\to0$.  With $\Omega$ as above, as resonances have $\Ima\le0$, each $\zeta\in\Omega$
has $\dist(\zeta,\Res P(h))\ge\min(w,D)=D$  by \eqref{rev:C1}, so (AP) with $g:=D/(2s(h))\le\frac12$ (again by \eqref{rev:C1}) gives $\|\chi(P(h)-\zeta)^{-1}\chi\|\le M:=\exp(\mathcal N\log\frac{4s}D)$ on
$\Omega$, while $\|\chi(P(h)-\zeta)^{-1}\chi\|\le1/\Ima\zeta=1/S$ on the top edge, by self-adjointness. For unit vectors $v,v'$ apply \eqref{eq:TZ-Lem2.2} 
to $F(\zeta):=\langle\chi(P(h)-\zeta)^{-1}\chi v,v'\rangle$: its hypotheses hold by \eqref{rev:C1} ($w/D\to\infty$) and \eqref{rev:C2} (with $C_1\ge 2C_0$, since
$\log(C_0MS)\le C_0+\mathcal N\log\frac{4s}D$), and it gives
\[
\|\chi(P(h)-\zeta)^{-1}\chi\|\le2e\,S^{-1}(MS)^{S/(D+S)}\le\frac{2e}S\exp\Big(\mathcal N\frac SD\log\frac{4s}D\Big)\qquad\text{on }[E\pm w].
\]
Since $(P(h)-E)u$ is supported in $\overline{B(0,R_0)}$ and $\chi(P(h)-\zeta)^{-1}\chi(P(h)-\zeta)u=u$ for $\Ima\zeta>0$ continues to $\zeta=E$, it follows that
$1=\|u\|\le\|\chi(P-E)^{-1}\chi\|\,R\le\frac12$ by \eqref{rev:C3}, a contradiction.
\end{proof}

\section{Tang-Zworski bound for $1\leq p<2$}

In the sequel, $\varphi(x)\in\mathbb R/2\pi\mathbb Z$ denotes the polar angle of $x\in\mathbb R^2\setminus\{0\}$, so that
$\mathbf F:=\nabla\varphi=(-x_2,x_1)/|x|^2$ is the unit-flux Aharonov--Bohm potential and $\mathbf A_\alpha=\alpha \mathbf F$; for $m\in\mathbb Z$ the function $e^{im\varphi}$
is smooth on $\mathbb R^2\setminus\{0\}$. Moreover $\theta\in(0,\pi/8)$; $f_\theta$ is the scaling function of Section 2.2 with $f_\theta(t)=t$ for
$t\le R_1$ and $f_\theta(t)=e^{i\theta}t$ for $t\ge T_0$; $P_\theta(h)$ acts on $L^2(\mathbb R^2)$ (identified with $L^2(\Gamma_\theta)$) with domain
$\mathcal D_\theta$; $[a,b]\subset(0,\infty)$; $\widetilde\Omega_0\subset\Sigma_\theta$ is a compact simply connected neighbourhood of $[a,b]$; and
$\chi\in C_c^\infty(\mathbb R^2)$ is radial with $\chi=1$ near $B(0,R_0)$. Constants $C,c$ never depend on $h$ or on $z$ in the indicated region.

By a scaling argument, we verify that the condition (AP) in Theorem~\ref{rev:threshold} holds in the case $s(h)=h^p,$ with $1\leq p<2$.

\begin{lemma}[A priori bound, $1\le p<2$]\label{rev:AP}
Let $1\le p<2$. There are $A$ and $h_1$ such that for $h\in(0,h_1]$, $g\in(0,\frac12]$,
\[
\big\|\chi(P(h)-z)^{-1}\chi\big\|\le\exp\Big(A\,h^{-2}\log\tfrac2g\Big)\qquad
\text{for }z\in h^p\widetilde\Omega_0\setminus\bigcup_{z_j\in\Res P(h)}D(z_j,h^pg);
\]
i.e.\ {\rm(AP)} holds with $s(h)=h^p$ and $\mathcal N(h)=Ah^{-2}$.
\end{lemma}

\begin{proof}
Let $\hat{h}:=h^{1-p/2}$, let $m=m(h)$ be the integer nearest to $\alpha/h$, $\beta:=\alpha/h-m\in[-\frac12,\frac12]$, $R_0^+:=R_0+1$, and let $V$ be
multiplication by the unimodular function equal to $1$ on $B(0,R_0^+)$ and equal to $e^{-im\varphi}$ on $\{|x|>R_0^+\}$. Put
$\widetilde P(\hat{h}):=h^{-p}VP(h)V^{-1}$, a self-adjoint operator with domain $VH^2(\mathbb R^2)$, and $\widetilde P_\theta(\hat{h}):=h^{-p}VP_\theta(h)V^{-1}$,
which is the complex scaling of $\widetilde P(\hat{h})$ since $V$ commutes with the radial deformation. For radial $\chi$ and $z=h^p\zeta$,
\begin{equation}\label{eq:scaleid}
\chi(P_\theta(h)-z)^{-1}\chi=h^{-p}V^{-1}\chi\big(\widetilde P_\theta(\hat{h})-\zeta\big)^{-1}\chi V .
\end{equation}
The domain of $\tilde P(\hat{h})$ is $VH^2$ and depends on $h$, whereas \cite{TZ,S} take it fixed. However, the proof of \cite[Lemma~1]{TZ} uses the domain only through the operator norms
appearing in Assumptions~\ref{ass:D-C}--\ref{Ass:Weyl}, which are all controlled uniformly in $\hat{h}$  (see Remark~\ref{rev:rem-scaling} below). Therefore, \cite[Lemma~1]{TZ} applies to $\widetilde P(\hat{h})$ on  a fixed compact simply connected
set $\widetilde\Omega\subset\Sigma_\theta$  containing $\{\zeta:\dist(\zeta,\widetilde\Omega_0)\le d_1\}$, where $d_1:=\frac12\dist(\widetilde\Omega_0,e^{-2i\theta}[0,\infty))>0$, and gives
$$
\|\chi(\widetilde P_\theta(\hat{h})-\zeta)^{-1}\chi\|\le A_0\exp\Big(A_0\hat{h}^{-\frac{4}{2-p}}\log\frac1g\Big)
\qquad\text{for }\zeta\in\widetilde\Omega_0\setminus\bigcup_jD(\zeta_j,g),
$$
the $\zeta_j$ running over the eigenvalues of $\widetilde P_\theta(\hat{h})$ in $\widetilde\Omega$, with $A_0$ independent of $\hat{h}$. For $g\le d_1$ every eigenvalue within distance $g$ of a point of $\widetilde\Omega_0$ lies in $\widetilde\Omega$, so the bound holds off the disks around all eigenvalues; for $g\in(d_1,\frac12]$ it follows from the case $g=d_1$ after enlarging $A_0$, since $\log(1/d_1)\le C\log(2/g)$ there. By \eqref{eq:scaleid} these
eigenvalues are $h^{-p}$ times the resonances of $P(h)$, and $\|\chi(P(h)-z)^{-1}\chi\|=h^{-p}\|\chi(\widetilde P_\theta(\hat{h})-\zeta)^{-1}\chi\|$
for $z=h^p\zeta$. Since $\hat{h}^{-\frac{4}{2-p}}=h^{-2}$ and $h^{-p}A_0\le e^{Ch^{-2}}$, the lemma follows with $A:=A_0+C$.
\end{proof}

\begin{remark}[Black box assumptions for $\widetilde P(\hat{h})$]\label{rev:rem-scaling}
    The operator $\widetilde P(\hat{h})$ satisfies Assumptions~\ref{ass:D-C}--\ref{Ass:Weyl} with semiclassical parameter $\hat{h}$, black-box radius $R_0^+=R_0+1$, exponent $h^{-n^{\sharp}}$ in the Weyl law (with $n^\sharp:=4/(2-p)$), and
constants independent of $\hat{h}$. We provide some details for the convenience of the reader. On $\{|x|>R_0^+\}$ one has $P(h)=h^p(-i\hat{h}\nabla-h^{-p/2}\alpha\mathbf F)^2$, and conjugation by $e^{-im\varphi}$, $\mathbf F=\nabla\varphi$,
replaces the potential by $(h^{-p/2}\alpha-\hat{h} m)\mathbf F=\beta\hat{h} \mathbf F$; hence $\widetilde P(\hat{h})=(-i\hat{h}\nabla-\beta\hat{h} \mathbf F)^2$ there, with first- and zeroth-order
coefficients $\mathcal O(\hat{h})$ and $\mathcal O(\hat{h}^2)$ in the normalization of Assumption~\ref{ass:ext-op}. Assumptions~\ref{Ass:C}--\ref{Ass:AC} then follow from those for $P(h)$ by unitary equivalence
and the computations of Section~\ref{sec:blackbox} with $(\alpha,h)$ replaced by $(\beta\hat{h},\hat{h})$. For the Weyl law (Assumption~\ref{Ass:Weyl}), the reference operator of $\widetilde P(\hat{h})$ may be
taken unitarily equivalent to $h^{-p}P^\sharp(h)$  (conjugate $P^\sharp(h)$ by a unimodular function on the torus equal to $V$ on $B(0,2R)$), so that
$N(\widetilde P^\sharp,[-\lambda,\lambda])=N(P^\sharp(h),[-\lambda h^p,\lambda h^p])\le C\lambda h^{-2}=C\lambda\hat{h}^{-4/(2-p)}$ for $\lambda\ge1$.
For Assumption~\ref{ass:D-C}, the only point requiring care is the uniform right inverse: the twist $e^{-im\varphi}$ of the domain $VH^2$ sits on the circle
$|x|=R_0^+$, inside the field-free annulus $R_0<|x|<R_0^+$ where $h^{-p}P(h)(e^{im\varphi}f)=e^{im\varphi}(-i\hat{h}\nabla-\beta\hat{h}\mathbf F)^2f$, so that for $w\in H^2(\{|x|>R_0^+\})$ the
extension $V(e^{im\varphi}\psi\bar w)$, with $\bar w$ an $H^2$-extension of $w$ to $\{|x|>R_0\}$ and $\psi$ a cutoff equal to $1$ on $\{|x|\ge R_0^+\}$ and
vanishing on $\{|x|\le R_0+\frac12\}$, lies in $VH^2$, restricts to $w$, and has graph norm $\le C\|\langle\hat{h} D\rangle^2w\|$ uniformly in $\hat{h}$.
\end{remark}

\section{Tang-Zworski bound for $p=2$}

In the proof of Theorem~\ref{thm:zero island}, we apply Theorem~\ref{rev:threshold} with $s(h)=h^2$ (hence $p=2$) under the following additional hypothesis
\begin{itemize}
\item[\textbf{(H)}] $\omega\subset\mathbb R^2$ is open, bounded, simply connected, with smooth boundary, $\overline\omega\subset B(0,R_0)$;
$\curl \mathbf A=0$ on $\omega$; and $\curl \mathbf A\ge c_0>0$ a.e.\ on $\mathfrak B:=B(0,R_0)\setminus\overline\omega$.
\end{itemize}
which allows us to establish validity of hypothesis (AP) (in Theorem~\ref{rev:threshold}) at this scale.

\subsubsection*{Notation}
 
Put $b:=1/h$ and $\cP:=h^{-2}P(h)=(-i\nabla-b\mathbf A)^2$, $\cP_\theta:=h^{-2}P_\theta(h)$ (the same operators up to the scalar $h^{-2}$). Then
$\Res P(h)=h^2\Res\cP$, $\cP=(-i\nabla-b\mathbf A_\alpha)^2$ on $\Omega_{\rm ext}:=\{|y|>R_0\}$, and for radial $\chi$ and $z=h^2\zeta$,
\begin{equation}\label{rev:cP-res}
\chi(P(h)-z)^{-1}\chi=h^{-2}\chi(\cP-\zeta)^{-1}\chi .
\end{equation}

On $\omega$, $\curl \mathbf A=0$ and $\omega$ is simply connected, so $b\mathbf A=\nabla\psi_b$ with $\psi_b\in H^2(\omega)$ real. The Dirichlet realization $\mathcal L_\omega$
of $(-i\nabla-b\mathbf A)^2$ on $\omega$ is $e^{i\psi_b}(-\Delta^D_\omega)e^{-i\psi_b}$, with eigenvalues $\ell_0\le\ell_1\le\dots$ (with multiplicity) and orthonormal
eigenfunctions $\phi_j:=e^{i\psi_b}\psi^D_j$, where $(\psi^D_j)$ are the Dirichlet eigenfunctions of $-\Delta$ on $\omega$.
On $\Omega_{\rm ext}$, define the unitary operator $U_b:=e^{-im\varphi}$ on $L^2(\Omega_{\rm ext})$, where $m$ is the integer nearest to $\alpha b$, and $\beta:=\alpha b-m\in[-\frac12,\frac12]$.
Note that $U_b^{-1}(-i\nabla-b\mathbf A_\alpha)^2U_b=(-i\nabla-\beta \mathbf F)^2$ on $\Omega_{\rm ext}$, where $\mathbf F=\nabla\varphi$, the Aharonov-Bohm potential with unit flux.\medskip
 
Fix $n$ and set $[a,b_*]:=[\frac12\ell_n,2\ell_n]$, where $\ell_n$ is the $n$'th Dirichlet eigenvalue in $\omega$. Also, put $L:=(b_*+2)/\sin\theta$, $\zeta_0:=\ell_n+iL$, $\Lambda:=4(2\ell_n+L+1)$,
$\widetilde\Omega':=[\frac18\ell_n,8\ell_n]+i[-\frac1{16}\ell_n\tan2\theta,\,L+3]$, a compact subset of $\Sigma_\theta$ used only to define $\delta_{\rm ext}$ in Lemma~\ref{rev:ext} below, and $\epsilon_0>0$ so small that the collars
$\mathfrak C_\omega:=\{y\in\mathfrak B:\dist(y,\omega)<\epsilon_0\}$ and $\mathfrak C_{R_0}:=\{R_0-\epsilon_0<|y|<R_0\}$ are disjoint.
With $\delta_{\rm ext}$ from Lemma~\ref{rev:ext} and $\delta_\theta:=\frac12\min(\delta_{\rm ext},\frac1{16}\ell_n\tan2\theta)$, we use three nested regions,
\begin{multline*}
\widetilde\Omega_0:=[\tfrac12\ell_n,2\ell_n]+i[-\tfrac14\delta_\theta,\,L+1]\ \subset\ 
\widetilde\Omega_1:=[\tfrac14\ell_n,4\ell_n]+i[-\tfrac12\delta_\theta,\,L+2]\ \\
\subset\ 
\widetilde\Omega_2:=[\tfrac18\ell_n,8\ell_n]+i[-\delta_\theta,\,L+3]\subset\widetilde\Omega',
\end{multline*}
and put $d_0:=\dist(\widetilde\Omega_0,\partial\widetilde\Omega_1)>0$. All three are compact, simply connected, contain $\zeta_0$, and lie in $\{|\zeta|\le\Lambda\}$. (AP) will be verified on $\widetilde\Omega_0$; the two larger regions are needed for the complex-analytic steps in Lemma~\ref{rev:apriori}.
 
\subsubsection*{Quasimodes}
 
With $\widehat\omega$, $L(b)$, $\widehat\ell_n(b)$, $\psi_n$, $u_n$ as in Section~\ref{sec:zero island}: $\supp u_n\subset\widehat\omega\subset B(0,R_0)$,
$\|u_n\|=1+O(e^{-c\sqrt b})$, and by Proposition 7.1
\begin{equation}\label{rev:qm}
\big\|(P(h)-h^2\widehat\ell_n(b))u_n\big\|\le Ce^{-c_0/\sqrt h};
\end{equation}
moreover $\widehat\ell_n(b)\to\ell_n$ \cite{GKS,KS}, so $\widehat\ell_n(b)\in[a,b_*]$ for $b$ large.
 
\subsubsection*{The barrier estimate}
 
The first step compares the form of the scaled operator with the magnetic form on the unscaled region. 
For $w$ supported in $B(0,R_1)$, write 
\[\begin{aligned}
    q_{\rm int}(w)&:=\int_{|y|<R_1}|(-i\nabla-b\mathbf A)w|^2\\
q^h_{\rm int}(w)&:=\int_{|y|<R_1}|(-ih\nabla-\mathbf A)w|^2\,dy=h^2q_{\rm int}(w).
\end{aligned}\]

The angle $\theta$ in the complex scaling is in $(0,\pi/8)$. Decomposing $\mathbb R^2$ we check by a routine argument that the real part of the operator $P_\theta(h)$ is bounded below by the form in $B(0,R_1)$.

\begin{lemma}[Coercivity of the scaled form]\label{rev:coercive}
There is $C_\theta$ such that for all $u\in\mathcal D_\theta$ and all $\eta\in C_c^\infty(B(0,R_1);[0,1])$,
\begin{equation}\label{rev:coer}
\Rea\big(e^{i\theta}\langle P_\theta(h)u,u\rangle\big)\ \ge\ \cos\theta\ q^h_{\rm int}(\eta u)-C_\theta\,h^2\,\|u\|^2 .
\end{equation}
\end{lemma}
 
\begin{proof}
Let $\eta^2+\eta_{\rm t}^2+\eta_{\rm f}^2=1$ with $\eta_{\rm t}$ supported in $\{R_0<|y|<T_0+1\}$ and $\eta_{\rm f}$ in $\{|y|>T_0\}$
(such a partition exists for the given $\eta\in C_c^\infty(B(0,R_1);[0,1])$). The IMS formula
$\langle P_\theta u,u\rangle=\sum\langle P_\theta\eta_\cdot u,\eta_\cdot u\rangle-\sum\langle h^2G_\theta\nabla\eta_\cdot\cdot\nabla\eta_\cdot\,u,u\rangle$
holds for the second-order operator $P_\theta$ with principal part $h^2\nabla\cdot G_\theta\nabla$, and the second sum is $\mathcal O(h^2)\|u\|^2$.
On the first piece $P_\theta=P(h)$, so $\langle P_\theta\eta u,\eta u\rangle=q^h_{\rm int}(\eta u)\ge0$ is real.
For the other two pieces, pull $P_\theta$ back to $y$-coordinates: the Aharonov--Bohm form $\mathbf A_\alpha\cdot dx=\alpha\,d\varphi$ is invariant under
the radial deformation $\kappa_\theta$, so with $L:=hD-\mathbf A_\alpha(y)$ (a symmetric operator, $\mathbf A_\alpha$ being real) one has
$P_\theta=L\cdot G_\theta L+hB\cdot L$ on $\{|y|\ge R_0\}$, where $G_\theta=(d\kappa_\theta)^{-1}(d\kappa_\theta)^{-T}$ is complex symmetric with
eigenvalues $1/f_\theta'(|y|)^2$, $|y|^2/f_\theta(|y|)^2$ in the sector $\{|\arg|\le2\theta+2\epsilon_0\}$, and $B$ is bounded (it comes from
derivatives of $d\kappa_\theta$ and vanishes for $|y|\ge T_0$). Hence for $w$ supported in $\{|y|\ge R_0\}$,
\[
\langle L\cdot G_\theta Lw,w\rangle=\int\sum_jg_j(y)|(Lw)_j|^2\,dy,\qquad
\Rea\big(e^{i\theta}\langle L\cdot G_\theta Lw,w\rangle\big)\ge\cos(\theta+2\epsilon_0)\,c\,\|Lw\|^2,
\]
in a frame diagonalizing $G_\theta$, and $|h\langle B\cdot Lw,w\rangle|\le Ch\|Lw\|\|w\|\le\frac c2\|Lw\|^2+Ch^2\|w\|^2$. Applying this to
$w=\eta_{\rm t}u$ and $w=\eta_{\rm f}u$ and summing gives \eqref{rev:coer}. (For $|y|\ge T_0$ one has exactly
$P_\theta=e^{-2i\theta}(hD-\mathbf A_\alpha)^2$ and $B=0$.)
\end{proof}
 
Outside $\omega$, we prove decay estimates as $b\to+\infty$. These are needed in Proposition~\ref{rev:uniform} below.
\begin{lemma}\label{rev:barrier}
Let $K$ be  a bounded operator on $L^2(\mathbb R^2)$ with $\|K\|\le2\Lambda$, $\zeta\in\widetilde\Omega_2$, $u\in\mathcal D_\theta$, and $v:=(\cP_\theta+K-\zeta)u$. Then
\begin{align}
q_{\rm int}(\eta u)&\le C\big(\|v\|\|u\|+\|u\|^2\big)\quad\text{for every }\eta\in C_c^\infty(B(0,R_1);[0,1]), \label{rev:qint}\\
b\,c_0\int_{\mathfrak B}|u|^2dx&\le C\big(\|v\|\|u\|+\|u\|^2\big), \label{rev:bar}\\
\|u\|_{L^2(\partial\omega)}+\|u\|_{L^2(\partial B(0,R_0))}&\le C\,b^{-1/4}\big(\|v\|+\|u\|\big). \label{rev:trace}
\end{align}
\end{lemma}
 
\begin{proof}
\eqref{rev:qint} follows from Lemma~\ref{rev:coercive}  divided by $h^2$, since $\Rea e^{i\theta}\langle(\cP_\theta+K-\zeta)u,u\rangle\le\|v\|\|u\|$ and $|\langle(K-\zeta)u,u\rangle|\le(2\Lambda+|\zeta|)\|u\|^2$. For \eqref{rev:bar} take $\eta=1$ on $B(0,R_0)$; the diamagnetic  inequality for
$(-i\nabla-b\mathbf A)^2$ gives $q_{\rm int}(\eta u)\ge b\int\curl \mathbf A\,|\eta u|^2\ge bc_0\int_{\mathfrak B}|u|^2$, since $\curl\mathbf A\ge0$ on $B(0,R_1)$ and $\curl\mathbf A\ge c_0$ on $\mathfrak B$ under (H).

It remains to establish \eqref{rev:trace}. Let $f\in H^1(\mathfrak C)$, $f\ge0$, on a collar $\mathfrak C=\mathfrak C_{\epsilon_0}$ of width $\epsilon_0$ adjacent to a smooth curve $\Gamma$. Then
$\int_\Gamma f^2\le C\big(\epsilon_0^{-1}\|f\|^2_{L^2(\mathfrak C)}+\|f\|_{L^2(\mathfrak C)}\|\nabla f\|_{L^2(\mathfrak C)}\big)$; apply it with $f=|u|$ on
$\mathfrak C_\omega$ and $\mathfrak C_{R_0}$, both contained in $\mathfrak B$, using $|\nabla|u||\le|(-i\nabla-b\mathbf A)u|$ a.e.\ together with \eqref{rev:qint}
(for $\eta=1$ on $B(0,R_0)$) and \eqref{rev:bar}.
\end{proof}
 
\subsubsection*{The exterior problem}
 
Let $H^{\rm ext}_\beta$ be the Dirichlet realization of $(-i\nabla-\beta \mathbf F)^2$ on $\Omega_{\rm ext}$ and $H^{\rm ext}_{\beta,\theta}$ its complex scaling
(with the same scaling function $f_\theta$). Here $\mathbf F=\mathbf A_1$, the Aharonov-Bohm potential with unit flux.
 
\begin{lemma}\label{rev:ext}
There is $\delta_{\rm ext}>0$ such that  no eigenvalue of $H^{\rm ext}_{\beta,\theta}$, for any $\beta\in[-\frac12,\frac12]$, lies in $\widetilde\Omega'\cap\{\Ima\zeta\ge-\delta_{\rm ext}\}$. Consequently, with $\widetilde\Omega_2$ as in the Notation,
$H^{\rm ext}_{\beta,\theta}-\zeta$ is invertible for all $\zeta\in\widetilde\Omega_2$, $\beta\in[-\frac12,\frac12]$, with $\|(H^{\rm ext}_{\beta,\theta}-\zeta)^{-1}\|\le C_{\rm ext}$.
\end{lemma}
 
\begin{proof}
Since $\mathbf F$ is smooth and bounded on $\Omega_{\rm ext}$, $\beta\mapsto H^{\rm ext}_{\beta,\theta}$ is an analytic family of type (A); for $\zeta\in\Sigma_\theta$,
$H^{\rm ext}_{\beta,\theta}-\zeta$ is Fredholm of index $0$ with discrete spectrum in $\Sigma_\theta$ \cite[\S3]{SZ}. No eigenvalue has $\Ima\zeta\ge0$: for $\Ima\zeta>0$
it would be a non-real eigenvalue of the self-adjoint $H^{\rm ext}_\beta$ (by \cite[Lemma~3.5]{SZ} and the $\theta$-independence of eigenvalues in
$\mathbb C\setminus e^{-2i[0,\theta]}[0,\infty)$), and for $\zeta>0$ an embedded eigenvalue, excluded by separation of variables as in Lemma~\ref{rev:noreal}.
The set $\{(\beta,\zeta)\in[-\frac12,\frac12]\times\widetilde\Omega':\zeta\text{ eigenvalue of }H^{\rm ext}_{\beta,\theta}\}$ is closed by continuity of the eigenvalues of an analytic family of type (A), and disjoint from the compact set $[-\frac12,\frac12]\times(\widetilde\Omega'\cap\{\Ima\zeta\ge0\})$, hence
at positive distance $\delta_{\rm ext}$ from it. The uniform bound follows from continuity of $(\beta,\zeta)\mapsto(H^{\rm ext}_{\beta,\theta}-\zeta)^{-1}$ on the
compact set $[-\frac12,\frac12]\times\widetilde\Omega_2$, where it is defined.
\end{proof}
 
\subsubsection*{Dirichlet problems with $L^2$ boundary data}
 
Let $N_\infty:=\#\{j:\ell_j\le3\Lambda\}$ and
\begin{equation}\label{rev:K}
K:=\sum_{j\le N_\infty}(2\Lambda-\ell_j)\,|\phi_j\rangle\langle\phi_j|\quad\text{on }L^2(\mathbb R^2),\ \phi_j\text{ extended by }0 .
\end{equation}
$K$ is self-adjoint, $\Rank K=N_\infty$, $\|K\|\le2\Lambda$, $Ku$ depends only on $u|_{\omega}$ being supported in $\overline\omega$, and
$\mathcal L_\omega+K$ (on $L^2(\omega)$) has spectrum $\{2\Lambda\}\cup\{\ell_j:j>N_\infty\}\subset[2\Lambda,\infty)$, at distance $\ge\Lambda$ from $\widetilde\Omega_2$.
 
\begin{lemma}\label{rev:dirichlet}
For $\zeta\in\widetilde\Omega_2$:
\begin{enumerate}
\item[\rm(a)] if $u\in H^2(\omega)$, $(\mathcal L_\omega+K-\zeta)u=f$ in $\omega$ (with $\mathcal L_\omega$ understood as the differential expression $(-i\nabla-b\mathbf A)^2$) and $u|_{\partial\omega}=g$, then
$\|u\|_{L^2(\omega)}\le C(\|f\|_{L^2(\omega)}+\|g\|_{L^2(\partial\omega)})$;
\item[\rm(b)] if $u\in L^2\cap H^2_{\rm loc}(\overline{\Omega_{\rm ext}})$ lies in the domain of $H^{\rm ext}_{\beta,\theta}$ away from the boundary,
$(H^{\rm ext}_{\beta,\theta}-\zeta)u=f$ in $\Omega_{\rm ext}$ and $u|_{\partial B(0,R_0)}=g$, then
$\|u\|_{L^2(\Omega_{\rm ext})}\le C(\|f\|+\|g\|_{L^2(\partial B(0,R_0))})$, with $C$ independent of $\beta$.
\end{enumerate}
\end{lemma}
 
\begin{proof}
We work after the gauges $e^{i\psi_b}$ in (a) and $U_b$ in (b), which preserve all the norms involved and turn the operators into $-\Delta+e^{-i\psi_b}Ke^{i\psi_b}-\zeta$ and $H^{\rm ext}_{\beta,\theta}-\zeta$, with coefficients independent of $b$.
Let $T$ be  the corresponding operator on its Dirichlet domain. For $\phi$ in the domain of $T^*$ (which is $H^2\cap H^1_0$
by elliptic regularity), Green's formula gives $\langle u,T^*\phi\rangle=\langle f,\phi\rangle-\langle g,\partial_\nu\phi\rangle_{L^2(\partial)}$. Given
$\psi\in L^2$, put $\phi:=(T^*)^{-1}\psi$; $T^*$ is invertible with $\|(T^*)^{-1}\|\le C$ (in (a) by the spectral gap, in (b) by Lemma~\ref{rev:ext}),
$\|\phi\|_{H^2}\le C\|\psi\|$ by $H^2$-regularity of the adjoint Dirichlet problem (in (b), uniformly in $\beta$, the scaling being supported in
$|y|\ge R_1>R_0$), and $\|\partial_\nu\phi\|_{L^2(\partial)}\le C\|\phi\|_{H^2}$. Hence $|\langle u,\psi\rangle|\le C(\|f\|+\|g\|)\|\psi\|$.
\end{proof}
 
\subsubsection*{The a priori bound}

With $K$ as in \eqref{rev:K}, we have:
 
\begin{proposition}\label{rev:uniform}
Let $\mathcal T(\zeta):=\cP_\theta+K-\zeta$. There is $b_2$ such that for $b\ge b_2$ and $\zeta\in\widetilde\Omega_2$, $\mathcal T(\zeta)$ is invertible and
$\|\mathcal T(\zeta)^{-1}\|\le3K_*$, with $K_*:=\max(\Lambda^{-1},C_{\rm ext})$.
\end{proposition}
 
\begin{proof}
$\mathcal T(\zeta)$ is a finite-rank perturbation of $\cP_\theta-\zeta$, Fredholm of index $0$ \cite[Lemma~5.1]{S}; it suffices to show $\|u\|\le3K_*\|\mathcal T(\zeta)u\|$.
Let $v:=\mathcal T(\zeta)u$. By Lemma~\ref{rev:barrier}, $\|u\|_{L^2(\mathfrak B)}\le Cb^{-1/2}(\|v\|+\|u\|)$ and the traces of $u$ on $\partial\omega$, $\partial B(0,R_0)$
are $\le Cb^{-1/4}(\|v\|+\|u\|)$. On $\omega$, $\cP_\theta=(-i\nabla-b\mathbf A)^2$ and $Ku=K(u|_\omega)$, so $u|_\omega$ solves the equation of Lemma~\ref{rev:dirichlet}(a)
with $f=v|_\omega$, whence $\|u\|_{L^2(\omega)}\le K_*\|v\|+Cb^{-1/4}(\|v\|+\|u\|)$. On $\Omega_{\rm ext}$, $\cP_\theta$ is the scaled $(-i\nabla-b\alpha \mathbf F)^2$ and $Ku=0$;
conjugating by $U_b$ (which preserves the norms of $u$, $f$, $g$), Lemma~\ref{rev:dirichlet}(b) gives $\|u\|_{L^2(\Omega_{\rm ext})}\le K_*\|v\|+Cb^{-1/4}(\|v\|+\|u\|)$.
Adding the three bounds, $\|u\|\le2K_*\|v\|+Cb^{-1/4}(\|v\|+\|u\|)$, whence $\|u\|\le3K_*\|v\|$ for $b\ge b_2$.
\end{proof}

Now, we are in  position to verify (AP) in the case $s(h)=h^2$ and under the hypothesis (H).
 
\begin{lemma}\label{rev:apriori}
There are $C$ and $b_3$ such that for $b\ge b_3$ and $g\in(0,\frac12]$,
\[
\|\chi(\cP-\zeta)^{-1}\chi\|\le C\,g^{-C}\qquad\text{for }\zeta\in\widetilde\Omega_0\setminus\bigcup_{\zeta_j\in\Res\cP}D(\zeta_j,g).
\]
Consequently {\rm(AP)} holds for $P(h)$ with $s(h)=h^2$ and $\mathcal N(h)=C(1+\log\frac1h)$.
\end{lemma}
 
\begin{proof}
Let $(e_j)_{j\le N_\infty}:=(\phi_j)$, an orthonormal basis of $\Ran K$, and $k_j:=2\Lambda-\ell_j$.
With $\mathcal T(\zeta)$ as in Proposition~\ref{rev:uniform}, we write 
\[\mathcal P_\theta-\zeta=\mathcal T(\zeta)\big(1-\mathcal T(\zeta)^{-1}K\big)\text{ on $\tilde\Omega_2$.}\]
We introduce the matrix of $\mathcal T(\zeta)^{-1}K$ in the basis $(e_j)$,
\[N(\zeta):=\big(k_j\langle\mathcal T(\zeta)^{-1}e_j,e_i\rangle\big)_{i,j\le N_\infty},\qquad d(\zeta):=\det\big(I-N(\zeta)\big).\]
Note that $\mathcal P_\theta-\zeta$ is invertible if and only if $1-\mathcal T(\zeta)^{-1}K$ is. Since $K$ has finite rank, this holds if and only if $d(\zeta)\ne0$. That said, the resonances in $\tilde\Omega_2$ are the zeros of $d$. By the Woodbury identity, 
\[(\mathcal P_\theta-\zeta)^{-1}=\mathcal T(\zeta)^{-1}+\mathcal T(\zeta)^{-1}\sum_{i,j}e_i\,[(I-N)^{-1}]_{ij}\,k_j\langle\mathcal T(\zeta)^{-1}\,\cdot\,,e_j\rangle,\]
we obtain 
\[\|\chi(\mathcal P_\theta-\zeta)^{-1}\chi\|\le C\big(1+\|\operatorname{adj}(I-N)\|/|d|\big)\le C(1+|d|^{-1}).\] 
Recall that $\zeta_0=\ell_n+iL$ with $L>0$. At $\zeta_0$, we have $(I-N(\zeta_0))^{-1}_{ij}=\delta_{ij}+k_j\langle(\mathcal P_\theta-\zeta_0)^{-1}e_j,e_i\rangle$. Since the functions $e_j$ are supported in $\overline\omega\subset B(0,R_0)$, by \cite[Lemma~3.5]{SZ} the last inner product equals $\langle\chi(\cP-\zeta_0)^{-1}\chi e_j,e_i\rangle$, of modulus $\le1/L$. Hence $\|(I-N(\zeta_0))^{-1}\|\le1+2\Lambda N_\infty/L$ and $|d(\zeta_0)|\ge C^{-1}$.

Jensen's formula on $\widetilde\Omega_2$ bounds the number $N'$ of zeros of $d$ in $\widetilde\Omega_1$ by a constant. Write $d(\zeta)=F(\zeta)\prod_j(\zeta-\zeta_j)$ over these zeros, where $F$ is holomorphic on $\widetilde\Omega_2$ and non-vanishing on $\widetilde\Omega_1$.

Since $|F|$ is bounded on $\partial\tilde\Omega_2,$ we get by the maximum principle  that $|F|$ is bounded on $\tilde\Omega_2$. In particular, there is a constant $C$ such that $|F|\le C$  on $\widetilde\Omega_1$.

We apply Harnack’s inequality  on the non-negative harmonic function $C-\log|F|$ on $\widetilde\Omega_1$. Since $|d(\zeta_0)|$ is bounded below by a positive constant, there is a positive constant $C_0<e^{C}$ such that $|F(\zeta_0)|\geq C_0$. With this bound in hand, Haranck's inequality gives  the uniform bound $C-\log|F|\leq C_1$ on $\widetilde\Omega_0$, for some constant $C_1$. Thus, letting $\tilde C=e^{C-C_1},$ we have
$|F|\geq \tilde C$ on $\widetilde\Omega_0$.

Let $\zeta\in\widetilde\Omega_0$ and $0<g\le d_0$. Suppose that $\zeta\notin D(\zeta_j,g)$ for every resonance $\zeta_j$ of $\cP$. Every resonance within distance $g$ of $\zeta$ lies in $\widetilde\Omega_1$, hence is one of the zeros $\zeta_j$ of $d$ in $\widetilde\Omega_1$. Consequently, $\prod_j|\zeta-\zeta_j|\ge g^{N'}$, and we get from the lower bound on $|F|$ that $|d(\zeta)|\ge\tilde Cg^{N'}$. Thus,
$$\|\chi(\cP-\zeta)^{-1}\chi\|=\|\chi(\cP_\theta-\zeta)^{-1}\chi\|\le Cg^{-N'}.$$
Finally, let $z=h^2\zeta$. Note that $z\notin D(z_j,h^2g)$ if and only if $\zeta\notin D(\zeta_j,g)$. Using \eqref{rev:cP-res}, we obtain
$$\|\chi(P(h)-z)^{-1}\chi\|\le Ch^{-2}g^{-N'}\le\exp\Big(\big(C'+2\log\tfrac1h\big)\log\tfrac2g\Big),$$
which is (AP) with $s(h)=h^2$ and $\mathcal N(h)=C'+2\log\frac1h$.
\end{proof}

\end{document}